\documentclass[11pt,reqno]{amsproc}

\usepackage[margin=1in]{geometry}
\usepackage{amsmath, amsthm, amssymb}
\usepackage{hyperref}
\usepackage{times}
\usepackage[usenames,dvipsnames]{color}

\usepackage{cjhebrew}

\title[Inviscid Limits for a Stochastic Shell Model]
{Inviscid Limits for a Stochastically Forced Shell Model of Turbulent Flow}

\author{Susan Friedlander}
\address{University of Southern California, Los Angeles, CA 90089}
\email{susanfri@usc.edu}

\author{Nathan Glatt-Holtz}
\address{Virginia Tech, Blacksburg, VA 24061}
\email{negh@vt.edu}

\author{Vlad Vicol}
\address{Princeton University, Princeton, NJ 08544}
\email{vvicol@math.princeton.edu}

\theoremstyle{plain}
\newtheorem{theorem}{Theorem}[section]
\newtheorem*{theorem*}{Theorem}
\newtheorem{definition}[theorem]{Definition}
\newtheorem{lemma}[theorem]{Lemma}
\newtheorem{proposition}[theorem]{Proposition}
\newtheorem{corollary}[theorem]{Corollary}

\theoremstyle{definition}
\newtheorem{remark}[theorem]{Remark}

\def\RR{{\mathbb R}}

\def\NN{{\mathbb N}}
\def\RR{{\mathbb R}}

\def\IC{\underline{u}}

\def\Prb{{\mathbb P}}
\def\E{{\mathbb E}}

\def\tilde{\widetilde}
\numberwithin{equation}{section}

\newcommand{\bfU}{\boldsymbol{u}}
\newcommand{\TT}{\mathbb{T}}
\newcommand{\DD}{\mathcal{D}}

\newcommand{\JJ}{\mathcal{J}}
\newcommand{\AAA}{\mathcal{A}}
\newcommand{\MM}{\mathcal{M}}
\newcommand{\MD}{\mathfrak{D}}
\newcommand{\Mspc}{\mathbb{D}}

\newcommand{\TaN}{\mathfrak{T}_{\alpha,N}}
\newcommand{\AddSet}{\textcjheb{k}}

\begin{document}

\begin{abstract}
We establish the anomalous mean dissipation rate of energy in the inviscid limit for a 
stochastic shell model of turbulent fluid flow.  
The proof relies on viscosity independent bounds for stationary solutions and on establishing 
ergodic and mixing properties for the viscous model.
The shell model is subject to a degenerate stochastic forcing 
in the sense that noise acts directly only through one wavenumber.
We show that it is hypo-elliptic (in the sense of H\"ormander) and use
this property to prove a gradient bound on the Markov semigroup.
\hfill \today
\end{abstract}

\maketitle

\setcounter{tocdepth}{1}
\tableofcontents

\section{Introduction}

Although there is a vast body of literature on Kolmogorov's theory
of turbulence, the dissipation anomaly, and the inviscid limit, at present
there is no rigorous mathematical proof that solutions to the Navier-Stokes equations yield Kolmogorov's laws.  
On the other hand, considering these
questions from a numerical perspective is costly and indeed in many 
situations lies beyond capacity of the most sophisticated computers.  For this reason researchers have  extensively investigated 
certain toy models, called {\em shell} or {\em dyadic} models, which are much simpler than the 
Navier-Stokes equations but which retain certain features of the nonlinear
structure. One such model was introduced by Desnianskii and Novikov~\cite{DesnianskiiNovikov1974},
to simulate the cascade process of energy transmission in turbulent flows. See also~\cite{FriedlanderPavlovic04,KiselevZlatos05,KatzPavlovic05,ConstantinLevantTiti07,MattinglySuidanVanden07,BessaihMillet2009,BessaihFlandoliTiti10,BarbatoFlandoliMorandin2010,Romito11,BarbatoFlandoliMorandin2012,Tao14}.
 
In this article we analyze statistically invariant states for the following {\em stochastically
driven} shell model of fluid turbulence.  For $j =0$ we take
\begin{align}
  d u_0 + (\nu u_0  + u_0 u_1) dt = \sigma dW
  \label{eq:shell:model:0}
\end{align}
where $W$ is a 1D Brownian motion
and $\sigma \in \RR$ measures the intensity of the noise. 
For $j \geq 1$
\begin{align}
  \frac{d}{dt} u_j + \nu 2^{2j}u_j +  (2^{cj} u_j u_{j+1} - 2^{c(j-1)} u_{j-1}^2) = 0.
    \label{eq:shell:model:j}
\end{align}
Here $\nu \geq 0$ and $c$ lies in the range $[1, 3]$.  

The main goal of the work is to establish that in the context of the stochastic dyadic model \eqref{eq:shell:model:0}--\eqref{eq:shell:model:j} some primary features of the 
Kolmogorov '41 theory of turbulence~\cite{Kolmogorov41a,Kolmogorov41b} hold. More precisely:
\begin{enumerate}
\renewcommand{\labelenumi}{(\Roman{enumi})}
\item In Theorem~\ref{thm:inviscid:limits} we prove that for $c\in[1,3]$, statistically stationary solutions $\bar u^\nu$ of the viscous shell model \eqref{eq:shell:model:0}--\eqref{eq:shell:model:j} converge as $\nu \to 0$ to statistically stationary solutions $\bar u$ of the inviscid shell model. Moreover, {\em the stationary inviscid solutions $\bar u$ experience 
an anomalous (or turbulent) dissipation of energy}: for any $N\geq 0$ we have a constant mean energy flux (cf. \eqref{eq:flux:derivation:ansatz} below)
\begin{align}
\E (\Pi_N(\bar u)) := \E (2^{cN} \bar u_N^2 \bar u_{N+1}) = \frac{\sigma^2}{2} = \epsilon > 0 .
\end{align} 
Moreover, we obtain that $\sup_{N\geq 0} 2^{2cN/3} \E |\bar u_N|^2 \leq C \epsilon^{2/3}$, where $C$ is a universal constant. This upper bound is consistent with the Kolmogorov spectrum, as described in Remark~\ref{rem:Kolmogorov} below.
\item In Theorem~\ref{thm:uniqueness:IM} we show that for $c\in[1,2)$, and any $\nu>0$, there exists a {\em unique invariant measure} for the Markov semigroup induced by \eqref{eq:shell:model:0}--\eqref{eq:shell:model:j} on the phase space $H = \ell^2$, {\em which is ergodic and exponentially mixing}.  Since \eqref{eq:shell:model:0}--\eqref{eq:shell:model:j} corresponds to a degenerate parabolic system, the main step in the proof relies on establishing that \eqref{eq:shell:model:0}--\eqref{eq:shell:model:j} is \emph{hypoelliptic} in the sense of H\"ormander. Here, the locality of the energy transfer in the nonlinear term complicates the bracket computations, and leads to a combinatorial problem. 
\item In Theorem~\ref{thm:dissipation:anomaly} we prove that for $c\in [1,2)$, {\em the mean dissipation rate of energy is bounded from below independently of viscosity}. 
More precisely there exists $\epsilon>0$ such that 
\begin{align}
\lim_{\nu \to 0} \lim_{T\to \infty} \frac{\nu}{T} \int_0^T  |u(t)|_{H^1}^2 dt = \frac{\sigma^2}{2} = \epsilon > 0
\label{eq:K41}
\end{align}
for every initial data $\{ u_j(0)\}_{j\geq 0}$ of finite energy, where the convergence occurs in an almost sure (pathwise) sense. In particular, 
the dissipation anomaly  $\sigma^2/2$ matches the inviscid anomalous energy dissipation rate.
\end{enumerate}

The manuscript is organized as follows.  
We begin our exposition with some further background from turbulence theory that motivate the
rigorous results established in Sections \ref{sec:math:setting}--\ref{sec:anomal:dissipation}.  
In Section~\ref{sec:math:setting} we briefly recall the mathematical setting of the stochastic shell model 
\eqref{eq:shell:model:0}--\eqref{eq:shell:model:j} and fix various mathematical notations used throughout.  Section~\ref{sec:moment:bnds:inviscid:limits}
is concerned with establishing $\nu$-independent bounds on statistically stationary solution of \eqref{eq:shell:model:0}--\eqref{eq:shell:model:j}.   We then
use these bounds to pass to a limit as $\nu \to 0$ and establish the existence of stationary solutions of the inviscid model.  We then show that these solutions 
exhibit a form  of turbulent dissipation.   As we already alluded to above, the results in Section~\ref{sec:moment:bnds:inviscid:limits} are valid over the entire range of $c$. 
In Section~\ref{sec:unique:attraction} we tackle the question of uniqueness, mixing and other attraction properties
for invariant measures of the viscous model in the more restricted range of $c \in [1,2)$.  The restriction $c<2$ implies that the equations are morally speaking semilinear, which allows us to obtain Foias-Prodi-type bounds.
The section concludes by  demonstrating that \eqref{eq:shell:model:0}--\eqref{eq:shell:model:j} satisfies
a form of the H\"ormander bracket condition. With this condition in hand the rest of the proof largely follows by using arguments similar to~\cite{HairerMattingly06, HairerMattingly2008,HairerMattingly2011, FoldesGlattHoltzRichardsThomann2013}.  Finally, Section~\ref{sec:anomal:dissipation} is devoted to proving
the dissipation anomaly \eqref{eq:K41}.  Appendices detail how a gradient bound on the Markov semigroup associated to 
 \eqref{eq:shell:model:0}--\eqref{eq:shell:model:j} can be derived from the H\"ormander bracket condition.
 We then show how various attraction properties for invariant measures may be established from these gradient bounds.

\section{Physical Motivation}
\label{sec:turbulence:back}
In this section we describe some further background concerning the Kolmogorov and Onsager theories of
turbulence which motivate the analysis of \eqref{eq:shell:model:0}--\eqref{eq:shell:model:j} carried out in this work.
\subsection{The Energy Flux, Dissipation Anomaly, and Anomalous Dissipation}
\label{sec:turbulence}
The motion of an inviscid, incompressible fluid is typically described by the Euler equations 
\begin{align}
   &\partial_t \bfU + (\bfU \cdot \nabla) \bfU = - \nabla p + f, \qquad \nabla \cdot \bfU = 0
   \label{eq:euler:eqn}
\end{align}
where $\bfU$ is the velocity field $p$ is the scalar pressure.
The viscous analogue of \eqref{eq:euler:eqn}, the Navier-Stokes equations, are given by
\begin{align}
   &\partial_t \bfU^\nu + (\bfU^\nu \cdot \nabla) \bfU^\nu = - \nabla p^\nu + \nu \Delta \bfU^\nu + f, \qquad \nabla \cdot \bfU^\nu = 0.
   \label{eq:NS:eqn}
\end{align}
Here $f$ is a (deterministic or random) force
which is frequency localized to act only at large scales of motion and $\nu$ is
the  kinematic viscosity coefficient of the fluid. The fluid domain  $\DD$ is either $\RR^3$ or $\TT^3$.   

Onsager \cite{Onsager1949} conjectured that every weak solution $\bfU$ to the Euler equations
with H\"older exponent $h > 1/3$ does not dissipate the kinetic energy $\int_{\DD} | \bfU|^2 dx$.  On the other
hand, the conjecture states that there exist weak solutions with smoothness 
less $h \leq 1/3$ which dissipate energy.  Such energy dissipation due to the
roughness of the flow is called {\em anomalous (or turbulent) dissipation}.  

The presence of energy dissipation in a viscous fluid with $\nu > 0$ is clear.
The mean energy dissipation rate per unit mass for an ensemble of solution
$\bfU^\nu$ to the Navier-Stokes equations \eqref{eq:NS:eqn}
is defined by 
\begin{align}
   \epsilon^\nu :=  \nu \langle \| \nabla \bfU^\nu\|_{L^2}^2 \rangle
   \label{eq:enegy:disp}
\end{align}
where the brackets $\langle \cdot \rangle$ denote a suitable average of the putative statistically steady state of \eqref{eq:NS:eqn}.\footnote{This 
operation $\langle \cdot \rangle$ is commonly defined as
a long time average made of the observable, which may be seen as 
an implicit invocation of an ergodic hypothesis: long-time averages
and averages against an invariant measure associated to the equations yield the same statistics.
While significant progress has been made on providing rigorous justification
for this hypothesis for the 2D stochastic NSEs it is completely
open in the three dimensional case.} It is a basic assumption of the classical theory of homogeneous, isotropic
turbulence proposed by Kolmogorov~\cite{Kolmogorov41a,Kolmogorov41b}
in 1941 that
\begin{align}
   \liminf_{\nu \to 0} \epsilon^\nu = \epsilon > 0.
   \label{eq:disp:anomaly}
\end{align}
The positivity of the energy dissipation rate in the limit of vanishing 
viscosity is called the {\em dissipation anomaly}. It is consistent with turbulence theory that the limiting value of $\epsilon$ is the dissipation rate due
to anomalous dissipation in the Euler equations.  There is an extensive
literature on these subjects and the connection between Onsager's conjecture
and Kolmogorov's hypothesis.  Several informative
reviews are given by~\cite{Frisch95,Robert03,EyinkSreeniviasan06}, which contain abundant 
references to the development of the topic over more than half a century.

The fundamental object of study in both the Onsager and Kolmogorov theories is the {\em energy flux}.
Formally, one may define the energy flux through the sphere of radius $2^j$ in frequency space as 
\begin{align}
  \Pi_j :=  \int_{\DD} \bfU \cdot \nabla S_j^2 \bfU \cdot \bfU dx,
  \label{eq:j:cut:off:flux}
\end{align}
where $\widehat{S_j \bfU} = \hat{\bfU} \psi(\cdot 2^{-j})$, and $\psi$ is a radial, smooth cut-off function centered
at the origin. The total energy flux is then given by
\begin{align}
  \Pi :=  \int_{\DD} (\bfU \cdot \nabla) \bfU \cdot \bfU dx = \lim_{j\to \infty} \Pi_j.
    \label{eq:total:flux}
\end{align}
The energy equation derived from \eqref{eq:euler:eqn} is
\begin{align}
  \frac{1}{2} \frac{d}{dt}  \int_{\DD} |\bfU|^2 dx = - \Pi + \int_{\DD} \bfU \cdot f dx.
  \label{eq:energy:balance:euler:form}
\end{align}
If $\bfU$ is sufficiently smooth, 
then since $\bfU$ is divergence free one may show that the energy flux vanishes. 
See~\cite{ConstantinETiti94} and more recently 
\cite{CheskidovConstantinFriedlanderShvydkoy08} for the sharper condition $\bfU \in B^{1/3}_{3,c_0}$ which ensures that $\Pi=0$.\footnote{Here the 
Besov space $B^{1/3}_{3, c_0}$ consists of functions such that $\lim_{j \to \infty} 2^j \|\bfU_j\|_{L^3}^3=0$.}
We note that to date there is no example of a weak solution to the Euler
equations in the Onsager critical space $B^{1/3}_{3, \infty}$ for which the 
energy flux $\Pi \not = 0$ and hence produces anomalous dissipation.\footnote{For a discussion 
of results concerning the existence of weak solutions
to the Euler equations, which experience anomalous dissipation
see~\cite{DeLellisSzekelyhidi13,Isett12,DeLellisSzekelyhidiBuckmaster13}, and references therein.}

An upshot of the proof in~\cite{CheskidovConstantinFriedlanderShvydkoy08} is that 
\begin{align}
  | \Pi_j| \leq C \sum_{i  =1}^\infty 2^{-2/3 |j - i|} 2^{i} \| \bfU_i\|_{L^3}^3
  \label{eq:LP:piece:est}
\end{align}
where $\bfU_i = (S_{i+1} - S_{i}) \bfU$ is the $i$th Littlewood-Paley piece
of $\bfU$. The estimate \eqref{eq:LP:piece:est} shows that energy transfer from
one scale to another is controlled mainly by \emph{local} interactions, which is one of the main motivations for considering the shell model~\eqref{eq:shell:model:0}--\eqref{eq:shell:model:j}, as we shall discuss below.

We now turn to the energy flux through wavenumber $2^{j}$ in the Navier-Stokes equations \eqref{eq:NS:eqn}, labeled $\Pi_j^\nu$. 
As in Kolmogorov's theory of turbulence, assume that the solutions $\bfU^\nu$ tend
to a statistically steady state, i.e. the statistical properties are independent of time and the solutions have
bounded mean energy, independently of $\nu$. In this case the average energy flux $\langle \Pi_j^\nu \rangle$ 
satisfies
\begin{align}
  \langle \Pi_j^\nu \rangle = - \nu \langle \| \nabla S_j \bfU^\nu \|_{L^2}^2 \rangle + \langle \int_{\DD} f \cdot S_j \bfU^\nu dx \rangle.
  \label{eq:weird:energy:val}
\end{align} 
In view of \eqref{eq:weird:energy:val}, upon passing $j \to \infty$ we obtain
\begin{align}
\nu \langle \| \nabla \bfU^\nu \|_{L^2}^2 \rangle = \lim_{j \to \infty} \nu \langle \| \nabla S_j \bfU^\nu \|_{L^2}^2 \rangle =  \lim_{j \to \infty} \langle \int_{\DD} f \cdot S_j \bfU^\nu dx \rangle - \lim_{j \to \infty} \langle \Pi_j^\nu \rangle = \langle \int_{\DD} f \cdot \bfU^\nu dx \rangle 
\end{align}
since $\bfU^\nu$ is sufficiently smooth for each fixed $\nu$. Thus, {\em assuming} that the Euler solution $\bfU$ is stationary in time, one would obtain as $\nu \to 0$
\begin{align}
\epsilon = \lim_{\nu \to 0} \epsilon^\nu = \lim_{\nu \to 0} \nu \langle \| \nabla \bfU^\nu \|_{L^2}^2 \rangle = \langle \int_{\DD} f \cdot \bfU dx \rangle  = \langle \Pi \rangle.
\label{eq:anomaly}
\end{align}
Here it is implicitly assumed that the turbulent statistically stationary solutions converge $\bfU^\nu \to \bfU$ in a certain averaged $L^2(\DD)$ sense. The energy flux thus provides the putative connection between the Kolmogorov and Onsager theories: the mean energy dissipation rate of turbulent stationary Euler solutions should match the vanishing viscosity limit of the mean energy dissipation rate in a turbulent stationary solution of the Navier-Stokes equation.
For further discussion of the connection between the Euler equations and turbulence see, for example~\cite{Frisch95,FoiasManleyRosaTemam01}, the recent articles~\cite{Shvydkoy09,CheskidovShvydkoy11,CheskidovShvydkoy12}, and references therein.

\subsection{Dyadic Models of Turbulent Flow}
\label{sec:dyadic:motivation}

Motivated by the Littlewood-Paley decomposition of the velocity field $\bfU=\sum_{j\geq 0} \bfU_j$, where $\bfU_j = (S_{j+1}-S_j)\bfU$, one may define the energy in the wavenumber shell $2^j \leq k \leq 2^{j+1}$ as $u_j^2 = \|\bfU_j\|_{L^2}^2$. In view of the locality of the energy transfer iterations implied by  \eqref{eq:LP:piece:est} one may thus define the flux through the shell at wavenumber $k=2^j$ as 
\begin{align}
  \Pi_j := 2^{cj} u_j^2 u_{j+1}
  \label{eq:flux:derivation:ansatz}
\end{align}
where $c$ is an ``intermittency parameter'' such that $1 \leq c \leq 5/2$.  The model
energy balance equation that mimics the Littlewood-Paley decomposition
of the Navier-Stokes equation thus becomes
\begin{align}
 \frac{1}{2} \frac{d}{dt}  u_j^2 =  - \Pi_j + \Pi_{j-1} - \nu 2^{2j} u_j^2 + f_j u_j
   \label{eq:energy:shell}
\end{align}
which upon substituting for $\Pi_j$ the formula \eqref{eq:flux:derivation:ansatz}, and setting the force to act only at the lowest wavenumbers, we obtain
our dyadic model given by the coupled system of ODEs for $\{u_j\}_{j \geq 0}$
\begin{align}
  &\frac{d}{dt} u_0 + \nu u_0  + u_0 u_1 = f_0,
  \label{eq:shell:model:0:det}\\
  &\frac{d}{dt} u_j + \nu 2^{2j}u_j +  (2^{cj} u_j u_{j+1} - 2^{c(j-1)} u_{j-1}^2) = 0, \quad j \geq 1.
    \label{eq:shell:model:j:det}
\end{align}
For a detailed discussion regarding the derivation of the shell model \eqref{eq:shell:model:0:det}--\eqref{eq:shell:model:j:det}, we refer the reader to~\cite{CheskidovFriedlanderPavlovic07,CheskidovFriedlander09,CheskidovFriedlanderPavlovic10}.

At this stage we would like to briefly  comment on the intermittency parameter $c$. The 1941-Kolmogorov theory of turbulence produces a power law for the energy density 
spectrum given by
\begin{align}
	{\mathcal E}(k) \sim \epsilon^{2/3} k^{-5/3},
	\label{eq:K:41:spectrum}
\end{align}
in the inertial range. This power law requires that velocity fluctuations are uniformly
distributed over the three dimensional domain $\DD$.  When taking into account that some spatial 
regions are more intensely turbulent than others, the power
laws become 
\begin{align}
{{\mathcal E}(k) \sim \epsilon^{2/3} k^{- \frac{8-D}{3}}}
\label{eq:spectrum:intermittent:case}
\end{align}
where $D$ is the Hausdorff dimension of the region of turbulent activity, and $\epsilon$ is redefined in terms of $D$, to have consistent units.
This phenomenon is referred to as spatial intermittency (see, for example~\cite{Frisch95,CheskidovShvydkoy12} and references therein).
On the other hand, the energy density spectrum ${\mathcal E}(2^j)$ associated with the {Onsager critical norm $H^{c/3}$ norm} is consistent with
\begin{align}
{2^{-j} \langle u_j^2 \rangle \sim {\mathcal E}(2^j) \sim \epsilon^{2/3} 2^{-j}2^{-\frac{2c}{3}j}}
\label{eq:dyadic:spectrum:c}
\end{align}
which yields, upon identifying $k=2^j$ that 
\begin{align}
 c = \frac{5- D}{2}.
 \label{eq:dyadic:turbulence:dim}
\end{align}
In particular, the range $1 \leq c < 2$
corresponds to $1 < D \leq 3$ with the end point $c =1$ corresponding to 
$D=3$ and the classical $k^{-5/3}$ power spectrum.  The range $2 \leq c \leq 5/2$
corresponds to $0 \leq D \leq 1$ where the regions of turbulence are
concentrated on thin sets that degenerate to points at the extreme value
$D=0$, $c =5/2$.  The analysis of the stochastic forced model that
we will present in this paper is strongly sensitive to the range of the parameter
$c$, as we will discuss in detail in the following sections.  

The properties of the system with a  constant force $f = (f_0, 0, \dots)$ and
$L^2$ initial data were established in~\cite{CheskidovFriedlanderPavlovic07,CheskidovFriedlander09,
CheskidovFriedlanderPavlovic10}.
It was shown that both in the inviscid and the viscous
model there is a unique fixed point which  is an exponential global attractor.
In the inviscid case this is achieved via anomalous dissipation.  Onsager's 
conjecture is verified in full with $H^{c/3}$ being the critical space.
It is proved that as $\nu \to 0$ the viscous global attractor converges  to the inviscid
fixed point.  Thus the average dissipation rate of the viscous
system converges to the anomalous dissipation rate $\epsilon$ of the inviscid
system.  Kolmogorov's theory is thus validated for the dyadic model 
\eqref{eq:shell:model:0:det}--\eqref{eq:shell:model:j:det} with a constant in time deterministic
force.

In this article we further adapt the dyadic model to the context of turbulence by studying a stochastically forced version.  
Stochastic shell models have also been considered
in a number of recent works, see e.g. \cite{BessaihMillet2009, BessaihFlandoliTiti10, BarbatoFlandoliMorandin2010, Romito11, BessaihFerrario2012, BarbatoFlandoliMorandin2012} and references therein. However, the model \eqref{eq:shell:model:0}--\eqref{eq:shell:model:j} considered here is perturbed by a highly degenerate frequency localized additive noise.
This degenerate situation has so far been addressed only for \emph{linear} shell models~\cite{MattinglySuidanVanden07}.  The
current work may therefore be seen as a continuation of \cite{MattinglySuidanVanden07} to a nonlinear context, inspired by some aspects of the Kolmogorov 1941 theory, which we describe next.

\subsection{Towards K41 for stochastic shell models}
As discussed above, the basic elements of the Komogorov '41 theory are:
\begin{itemize}
\item[(i)] For each $\nu>0$ and any initial data $u^\nu_0$, as $t \to \infty$ the corresponding solution $u^\nu(t)$  approaches a unique statistically steady state $\bar u^\nu$. 
\item[(ii)] There exists $\epsilon>0$ such that the statistically stationary solutions $\bar u^\nu$ obey $\lim_{\nu\to0} \nu  \langle |\nabla \bar  u^\nu|^2 \rangle  \geq \epsilon$.
\item[(iii)] The family $\{ \bar u^\nu \}_{\nu>0}$ is compact in the associated class of probability measures, and  along subsequences it converges to a statistically stationary solution $\bar u$ of the forced Euler equations.  These stationary Euler solutions experience a constant mean energy dissipation rate which is the same as for the viscous equations, namely $\epsilon > 0$.
\end{itemize}
Proving (i)--(iii) directly from the Navier-Stokes equations, remains an outstanding open problem.

One common setting for studying (i)--(iii) is to consider a wave-number localized, gaussian and white in 
time forcing to the governing equations. This serves as a proxy for generic large scale processes 
driving turbulent cascades.  The stochastic framework has been used extensively both theoretically and numerically~\cite{Novikov1965, BensoussanTemam, VishikKomechFusikov1979, Eyink96,EyinkSreeniviasan06, HairerMattingly06} and references therein.
Here one may take advantage of the
tools and techniques of  stochastic analysis in a regime where the injection of noise does not 
wash out the intricate underlying deterministic dynamics of the Navier-Stokes and Euler equations.  
In this setting invariant measures, i.e. statistically invariant states, are expected to encode
the statistics of turbulent flow at high Reynolds number.

Progress towards establishing (i) and (ii) has so far occurred in settings which are far from the 3D Navier-Stokes equations.
The uniqueness and attracting properties of the invariant measure for the 2D stochastic Navier-Stokes equations on the torus has recently 
been established e.g. in~\cite{HairerMattingly06, HairerMattingly2011}.\footnote{Note that in the two-dimensional case, instead of $\epsilon$, in (ii)  
one should consider $\eta$ the mean enstrophy dissipation rate.}  
We emphasize however that if the amplitude of the noise does not vanish in the inviscid limit, the sequence of Navier-Stokes stationary solutions does not converge as 
$\nu \to 0$, in any norm whatsoever~\cite{KuksinShirikian12}. In particular, (iii) does hold here.\footnote{The tightness of the Navier-Stokes invariant measures when the noise scales as $
\sqrt{\nu}$ has been addressed e.g.~in~\cite{KuksinShirikian12, GlattHoltzSverakVicol2013}. These solutions however do not obey the Batchelor-Kraichnan spectrum. On the other hand the 
convergence (iii), has been proven in the setting of the 1D stochastic Burgers equations~\cite{EKhaninSinai00}.  This work makes fundamental use of explicit representations of  solutions 
through the Lax-Oleinik formula and furthermore subjects the equations to a space-time white noise.}  
This is one of the main differences between the main conclusions (Theorems~\ref{thm:inviscid:limits}, \ref{thm:uniqueness:IM}, and \ref{thm:dissipation:anomaly}) of our work and the results for 
the 2D stochastic Navier-Stokes equations: not only do our viscous solutions 
obey a $\nu$-independent energy dissipation rate, but they also converge as $\nu \to 0$ to the solutions of the corresponding inviscid model. Moreover the inviscid stationary solutions 
experience turbulent dissipation due to a non-vanishing energy flux.\footnote{Another situation where an inviscid stochastic dyadic model has been shown to evidence dissipative behavior is 
developed in~\cite{BarbatoFlandoliMorandin2010,BarbatoFlandoliMorandin2012}. However, here randomness enters the equations as a formally conservative multiplicative Stratonovich noise.}

\section{Mathematical Setting and Preliminaries}
\label{sec:math:setting}
In this section we set the mathematical framework that will be used throughout the manuscript.

\subsection{Functional Setting}

We begin by recalling various sequence space based analogues of the classical Sobolev spaces. 
We denote the $\ell^2$-type sequence spaces by
\begin{align*}
  H^\alpha := \Big\{ u \in \ell^2(\NN):  |u|_{H^{\alpha}}^2 =\sum_{j \geq 0} 2^{2 \alpha j} u_j^2 < \infty \Big\}
\end{align*}
and define $\ell^\infty$-based sequence spaces (the replacement of the usual Lipschitz classes) by
\begin{align*}
  W^{\alpha,\infty} := \Big\{ u \in \ell^\infty(\NN):  |u|_{W^{\alpha,\infty}} = \sup_{j \geq 0} 2^{ \alpha j} |u_j| < \infty \Big\},
  \quad W^{\alpha, \infty}_{c_0} := \Big\{ u \in W^{\alpha,\infty}: \lim_{j \to \infty} 2^{ \alpha j} |u_j|  =0\Big\}.
\end{align*}
Observe that $H^1\subset W^{\alpha,\infty}$ with continuous embedding for $\alpha \leq 1$. We shall denote $H^0$ simply by $H$, and the norm associated to $\alpha = 0$ by $|\cdot|$.   
Finally, since we will often restrict our attention to solutions which are ``positive'' (away from the directly forced zeroth component), we take
\begin{align}
 H_{+} = \{ u \in \ell^2: u_j \geq 0, j \geq 1\}
\end{align}
and note that $H_{+}$ is a closed subset of $H$.

We define the operators
\begin{align}
  A u = ( 2^{2j} u_j )_{j \geq 0}, \quad B(u, v)  =  (2^{cj}  u_{j} v_{j+1} - 2^{c(j-1)} u_{j-1} v_{j-1})_{j \geq 0}.
\end{align}
Here and throughout the paper we use the convention that $u_{-1} = v_{-1} = 0$.  
We denote by $P_N u$ the projection of $u$ onto its first $N+1$ coordinates, i.e. $P_N u = (u_j)_{0\leq j \leq N}$.
Regarding the bilinear operator $B$ observe that for $u \in H^{c-1}$, $v \in H^1$ and $w \in H$
\begin{align}
  |\langle B(u, v), w \rangle| =& \sum_{j \geq 0} \left( 2^{cj}  |u_{j} v_{j+1} w_j| +  2^{c(j-1)} |u_{j-1}  v_{j-1} w_j| \right)
  \notag\\
  \leq& C \Big(\sup_{j \geq 0} 2^j |v_j|\Big) \Big(\sum_{j \geq 0} 2^{2(c-1)j} u_{j}^2\Big)^{1/2} \Big(\sum_{j \geq 0}w_j^2\Big)^{1/2}
  \leq C |u|_{H^{c -1}} |v|_{H^1} |w|.
  \label{eq:B:bnds}
\end{align}
As such, we have the cancelation property for $u, v \in H^{c-1}$,
\begin{align}
  \langle B(u, v), v \rangle = \sum_{j \geq 0} (2^{cj}  u_{j} v_{j+1} v_j - 2^{c(j-1)} u_{j-1} v_j v_{j-1})  = 0.
    \label{eq:B:cancelation}
\end{align}
In fact this can be improved to $u, v \in W^{c/3,\infty}_{c_0} \supset H^1$ when $c\leq 3$. With this formalism we may now rewrite \eqref{eq:shell:model:0}--\eqref{eq:shell:model:j} in the more abstract notation
which will sometimes serve as a useful shorthand:
\begin{align}
  du + (\nu Au + B(u,u))dt = e_0 dW, \quad u(0) = \IC.
\label{eq:abstract:shell}
\end{align}
To make the notion of solution rigorous, we next  recall some well-posedness properties.

\subsection{Existence and Uniqueness of Solutions}

The existence and uniqueness of solutions of \eqref{eq:shell:model:0}--\eqref{eq:shell:model:j}
is recalled in the following proposition which is essentially due to \cite{Romito11} and follows along the lines
of \cite{CIME08} (see also the related works
\cite{ConstantinLevantTiti07, BarbatoFlandoliMorandin2010, BarbatoMorandinRomito11}).

\begin{proposition}[\bf Existence and uniqueness of solutions, statistically steady states]
\label{prop:well:poshness}
Fix $\nu > 0$ and any $\IC \in H$.  
\begin{itemize}
\item[(i)] When $c \in [1, 3]$ there exists a martingale solution $(u, \mathcal{S})$ solving \eqref{eq:shell:model:0}--\eqref{eq:shell:model:j}
relative to the initial condition $\IC$ with the regularity
\begin{align}
   u \in L^2(\Omega; L^\infty_{loc}( [0,\infty); H) \cap L^2_{loc}([0,\infty); H^1)), \quad u_j \in C([0, \infty)) \textrm{ a.s. for each } j \geq 0.
   \label{eq:weak:sol:reg}
\end{align}
Here $\mathcal{S}= (\Omega, \mathcal{F}, \{\mathcal{F}_t\}, \Prb,W)$ is a stochastic basis which is considered as an unknown in the problem.
\item[(ii)] If $\IC \in H_+$ then, for any martingale solution $(u, \mathcal{S})$, $u(t) \in H_+$ for every $t \geq 0$.  Moreover, the solution 
$(u, \mathcal{S})$ can be chosen in such a way that the following moment bounds hold
\begin{align}
\E |u(t)|^2 + 2 \nu \int_0^t \E |u(s)|_{H^1}^2 ds \leq |\IC|^2 + t \sigma^2,
\label{eq:mean:energy:bal}
\end{align}
and for any $\kappa < \frac{\nu}{8 \sigma^2}$
\begin{align}
  \E \exp\left( \kappa \left( |u(t)|^2 +  \exp\left(-\frac{\nu t}{2} \right)  \int_0^t |u(s)|_{H^1}^2 ds \right) \right)
  	\leq \exp\left( \frac 14 + \kappa e^{-\frac{\nu t}{2}} | \IC |^2 \right).
	\label{eq:exp:moment:decay:bnd}
\end{align}
\item[(iii)] For every $\nu >0$, $c \in [1,3]$ there exists a stationary martingale solution $(\bar{u}^\nu, \mathcal{S})$ of the dyadic model;
there exist a stochastic basis $\mathcal{S}$ and time stationary process $\bar{u}^\nu$ with the regularity $\eqref{eq:weak:sol:reg}$
and solving \eqref{eq:shell:model:0}--\eqref{eq:shell:model:j}.  Moreover $(\bar{u}^\nu, \mathcal{S})$ can be chosen so that
\begin{align}
  \bar{u}^\nu \in H^{+}, \textrm{ a.s. }
    \label{eq:SSS:positivity}
\end{align}
to so as to satisfy the moment bound
\begin{align}
  \E\exp( \kappa |\bar{u}^\nu|^2) \leq \exp(1/4)
	\label{eq:SSS:exp:moment}
\end{align}
valid for any $\kappa < \frac{\nu}{8 \sigma^2}$.
\item[(iv)] In the case when $c \in [1,2]$ we may fix a stochastic basis $\mathcal{S}= (\Omega, \mathcal{F}, \{\mathcal{F}_t\}, \Prb, W)$.
Then, there exists a unique (pathwise) solution $u = u(\cdot, u_0, W)$ satisfying  \eqref{eq:shell:model:0}--\eqref{eq:shell:model:j}
and which has the regularity \eqref{eq:weak:sol:reg}.  Moreover $u(t, u_0, W)$ satisfies \eqref{eq:mean:energy:bal} with 
an equality and depends continuously on both $u_0$ in $H$ and 
on $W \in C([0,T])$.
\end{itemize}
\end{proposition}

The proof of Proposition~\ref{prop:well:poshness} is somewhat technical but represents a standard application 
of existing techniques.  For brevity we omit complete details, sketching only the main points.  
For the existence of Martingale solutions, (i) the proof follows 
precisely along the line of \cite{CIME08} using
compactness arguments around a Galerkin approximation of \eqref{eq:shell:model:0}--\eqref{eq:shell:model:j}
and variants of the Aubin-Lions and Arzela-Ascoli compactness
theorems.  Passage to the limit is facilitated Skorokhod embedding and by a Martingale representation theorem
from \cite{ZabczykDaPrato1992}, or alternatively by including the driving noise in the compact sequence
(see \cite{Bensoussan1995} or more recently \cite{DebusscheGlattHoltzTemam1}).

For the desired properties in (ii) observe that for $\IC \in H_{+}$ applying the Duhamel principle 
to \eqref{eq:shell:model:j} for each $j\geq 1$, gives
\begin{align}
  u_j(t) =& \exp\left( - \nu 2^{2j}t + 2^{cj} \int_0^t u_{j+1} ds \right) \IC_j \notag\\
  	     &+ \int_0^t \exp\left( - \nu 2^{2j}(t-s) + 2^{cj} \int_s^t u_{j+1} dr \right) u_{j-1}^2 ds.
	     \label{eq:positivity:via:duhammel}
\end{align}
The moment estimates \eqref{eq:mean:energy:bal}, \eqref{eq:mean:energy:bal} 
are formally identical to well known moment estimates for the stochastic Navier-Stokes equations
(cf. \cite{HairerMattingly06, Debussche2013,KuksinShirikian12}).

The existence of stationary solutions in (iii) follows
from a Krylov-Bogolyubov averaging procedure, implemented at the level
or Galerkin approximations.
Regarding the positivity of $\bar{u}$, \eqref{eq:SSS:positivity}, by choosing $\IC \in H^+$
for the Krylov-Bogolyubov averaged measure $\mu_T$ we infer from \eqref{eq:positivity:via:duhammel}
that $\mu_T(H_+) =1$. Then since $H_+$ is closed $\mu(H_{+}) \geq \limsup_j \mu_{T_j} (H_{+}) =1$.  
The moment bounds, \eqref{eq:SSS:exp:moment} are inferred from 
\eqref{eq:exp:moment:decay:bnd} via standard argument making use of invariance and decay 
of initial conditions evident in \eqref{eq:exp:moment:decay:bnd}.  See, for instance, 
\cite{Debussche2013, KuksinShirikian12}.

Regarding (iv) and the existence and uniqueness of pathwise solutions, since we are in the case of an additive noise, we can transform \eqref{eq:shell:model:0}
to a random process as follows:  Consider the Ornstein--Uhlenbeck process $dz_0 + \nu z_0 = \sigma dW$, $z(0) = 0$ and take 
$\tilde{u} = u - z e_0$.  Then $\tilde{u}$ solves
\begin{align}
  &\frac{d}{dt} \tilde{u}_0 + \nu \tilde{u}_0 + (\tilde{u}_0 + z_0) u_1 = 0, 
  \label{eq:shifted:system:0}\\
  &\frac{d}{dt} \tilde{u}_1 + \nu 2^{2}\tilde{u}_1 +  2^{c} \tilde{u}_1 \tilde{u}_{2} - \tilde{u}_{0}^2 = 2 z_0 \tilde{u}_0 + z_0^2
  \label{eq:shifted:system:1}\\
  &\frac{d}{dt} \tilde{u}_j + \nu 2^{2j}\tilde{u}_j +  2^{cj} \tilde{u}_j \tilde{u}_{j+1} - 2^{c(j-1)} \tilde{u}_{j-1}^2 = 0, \quad j \geq 2.
  \label{eq:shifted:system:j}
\end{align}
With this transformation in hand we can then implement a Galerkin approximation procedure 
for the associated transformed system.  The necessary compactness to pass to the the limit 
can then be treated pathwise.  To show that the limiting object $u = \tilde{u} + z$ is suitably 
adapted to the given filtration one also shows that \eqref{eq:shifted:system:0}--\eqref{eq:shifted:system:j}
depends continuously on $z$.

The continuous dependence of solutions on data can be established for $c \in [1,2]$ in a direct fashion as follows:  
Suppose that $u^{(1)}, u^{(2)}$  are solutions  of \eqref{eq:weak:sol:reg} (relative to the same stochastic basis) and let $v  = u^{(1)} - u^{(2)}$.  We have that $v$ satisfies
 $
   \frac{d}{dt} v + A v + B(v, u^{(1)}) + B(u^{(2)}, v) = 0
$.   
 Since $v \in L^2(\Omega; L^2_{loc}([0,\infty); H^1))$ we can make use of \eqref{eq:B:cancelation} and \eqref{eq:B:bnds} to 
 infer
$
  \frac{1}{2} \frac{d}{dt} |v|^2 + |v|^2_{H^1} \leq C |u^{(1)}|_{H^1} |v| |v|_{H^1} 
$.
 With $\epsilon$-Young and the Gr\"onwall inequality we infer
 \begin{align}
	|v(t)|^2 \leq |\underline{v}|^2 \exp\left(C\int_0^t |u^{(1)}|_{H^1}^2\right)
	\label{eq:con:dep:est}
 \end{align} 
 Uniqueness of solutions and continuous dependence on initial conditions follows. When $c>2$, the equation is quasi-linear and establishing the continuous dependence on data in the topology of $H$ seems out of reach.

\section{Uniform Moment Bounds and Inviscid Limits}
\label{sec:moment:bnds:inviscid:limits}

In this section we establish a series of $\nu$-independent
moment bounds for statistically stationary 
states of \eqref{eq:shell:model:0}--\eqref{eq:shell:model:j}.
Note carefully that the forthcoming bounds are valid for $c \in [1, 3]$.
These bounds allow us to pass to inviscid limit in this class 
of statistically invariant states and hence to establish the existence 
of stationary solutions of the inviscid model, 
that is \eqref{eq:shell:model:0}--\eqref{eq:shell:model:j} with $\nu = 0$.  
Such solutions are evidence of a form of turbulent dissipation as we detail below.
The $\nu$ independent moment bounds we establish are:
\begin{proposition}[\bf $\nu$-Independent moment bounds]
\label{prop:uniform:moment:bnds}
For each $\nu >0$ consider a stationary martingale solution $(\bar{u}^\nu, \mathcal{S})$ as in
Proposition~\ref{prop:well:poshness}, satisfying the positivity condition \eqref{eq:SSS:positivity},
and moment bound \eqref{eq:SSS:exp:moment}.  Then
\begin{align}
\sup_{\nu \in(0,1]}  \sup_{j\geq 0} 2^{(c-1) j} \E \left( (\bar u^\nu_j)^2\right)  < \infty
  \label{eq:nu:unif:BDD}
\end{align}
and moreover we have
\begin{align}
  \sup_{ \nu \in (0,1]} \E |\bar{u}^\nu|_{H^{a}}^2 < \infty
  \label{eq:nu:unif:Ha}
\end{align}
for each
$
   -1 \leq a < (c-1)/2,
$
when $c\in[1,3]$.
\end{proposition}
\noindent In particular, for any $c \in [1,3]$ the above proposition implies
\begin{align}
  \sup_{\nu \in(0,1]} \E |\bar{u}^\nu|^2_{H^{-1/2}} < \infty.
  \label{eq:nu:unif:H}
\end{align}

Working from the uniform bounds \eqref{eq:nu:unif:H} we are able to derive
the existence of stationary solutions $\bar u$ of the inviscid counterpart of the dyadic model
\eqref{eq:shell:model:0}--\eqref{eq:shell:model:j} namely
\begin{align}
  &d \bar u_0 +  \bar u_0 \bar u_1 dt = \sigma dW
  \label{eq:shell:model:0:inviscid}\\
  &\frac{d  \bar u_j}{dt} +  (2^{cj} \bar u_j \bar u_{j+1} - 2^{c(j-1)} \bar u_{j-1}^2) = 0, \quad j \geq 1
    \label{eq:shell:model:j:inviscid}
\end{align}
Motivated by the discussion in Section~\ref{sec:turbulence:back}, we define the \emph{energy 
flux through the $N$th shell} by
\begin{align}
  \Pi_N( u) := \langle P_N B( u, u), P_N  u \rangle = 2^{cN}  u_N^2  u_{N+1}
  \label{eq:flux:def}
\end{align}
for any $u \in H$.  We will see that statistically stationary solutions of 
\eqref{eq:shell:model:j:inviscid} must exhibit a constant average flux
independent of $N$. Our results concerning \eqref{eq:shell:model:0:inviscid}--\eqref{eq:shell:model:j:inviscid} 
are summarized as follows:

\begin{theorem}[\bf Stationary solutions of the Inviscid dyadic model]
\label{thm:inviscid:limits}
There exists a stationary martingale solution $(\bar{u}, \mathcal{S})$ of \eqref{eq:shell:model:0:inviscid}--\eqref{eq:shell:model:j:inviscid}
which satisfies the regularity
\begin{align*}
  \bar{u} \in L^{\infty}_{loc}([0,\infty); H^{a}), \quad \bar{u}_N \in C([0,\infty)) \mbox{ for each } N \geq 0, \qquad a.s.
\end{align*}
for any $a < c/3$.  Also, we have that the moment estimate
\begin{align}
\sup_{N \geq 0} 2^{2cN /3}\E ( \bar u_N^2) \leq C \sigma^{4/3}
\label{eq:best:quadratic:moments}
\end{align} 
holds, where $C>0$ is a universal constant. 
Furthermore, 
\begin{itemize}
\item[(i)]  Such solutions $\bar{u}$ may be obtained as an inviscid limit, namely,
there exists Borel probability measures $\{\mu_{\nu_j}\}$ and $\mu_0$ on $H$ such that
\begin{align}
  \mu_{\nu_j} \rightharpoonup \mu_0 \textrm{ in }  H^{-1/2} \textrm{ as } \nu_j \to 0
\end{align}
where $\mu_{\nu_j}(\cdot) = \Prb( \bar{u}^{\nu_j} \in \cdot)$ with $\bar{u}^\nu$ stationary solutions of \eqref{eq:shell:model:0}--\eqref{eq:shell:model:j}
and $\mu_0(\cdot) = \Prb( \bar{u} \in \cdot)$.
\item[(ii)] \label{thm:inviscid:limits:ii}
These inviscid stationary solutions $\bar u$ 
have a constant mean energy flux, i.e.
\begin{align}
   \E (2^{cN} \bar{u}_N^2 \bar{u}_{N+1}) = \E \Pi_N(\bar{u})  = \frac{\sigma^2}{2}
   \label{eq:flux:bound}
\end{align}
holds for any $N \geq 0$.
In particular we infer that
\begin{align}
\lim_{N \to \infty} 2^{cN} \E |\bar u_N|^3 > 0.
\label{eq:Onsager:moment}
\end{align}
\end{itemize}
\end{theorem}

\begin{remark}[\bf Consistency with Kolmogorov and Onsager]
\label{rem:Kolmogorov}
In view of \eqref{eq:flux:bound}  the constant mean energy flux is $\epsilon = \sigma^2/2$, so that $\epsilon^{2/3} \sim \sigma^{4/3}$. As such, the estimate \eqref{eq:best:quadratic:moments} is an upper bound consistent with the Kolmogorov power spectrum, in the case $c=1$, as described in 
\eqref{eq:dyadic:spectrum:c} above. Additionally, \eqref{eq:Onsager:moment} indicates that the inviscid steady state $\bar u$ has regularity below the Onsager critical space.
\end{remark}

\subsection{Uniform in $\nu$ Bounds}
\label{sec:Uniform:nu:bound}

Take $\{\bar{u}^{\nu}\}_{\nu > 0}$ to be statistically stationary solutions of  
\eqref{eq:shell:model:0}--\eqref{eq:shell:model:j} whose existence follows
from the Krylov-Bogolyubov and a possible usage of Galerkin approximations
with an appropriate limiting procedure.\footnote{In the case that $c \in [1,2]$ these
stationary solutions are unique and correspond to the (mixing) invariant measures $\{\mu_\nu\}_{\nu > 0}$
studied below in Section~\ref{sec:unique:attraction}.  These additional uniqueness
properties will have no bearing for the results in this section.}  As we explain in Section~\ref{sec:math:setting},
we can choose these elements $\bar{u}^\nu$ so that $\bar{u}^\nu \in H_+$.   We will
make crucial use of this positivity condition in the forthcoming computations.

Working from \eqref{eq:shell:model:0}--\eqref{eq:shell:model:j} and using stationarity we immediately have that, 
\begin{align}
  \nu 2^{2j} \E\left( \bar{u}_j^\nu \right) + 2^{cj} \E \left(  \bar{u}_j^\nu \bar{u}_{j+1}^\nu \right) = 2^{c(j-1)} \E \left( (\bar{u}^\nu_{j-1})^2 \right),
  \label{eq:stat:bal:0}
\end{align}
which holds for each $j \geq 0$. Here we are maintaining the convention that $\bar{u}^\nu_{-1} \equiv 0$.
Applying the It\={o} lemma to \eqref{eq:shell:model:0}--\eqref{eq:shell:model:j} 
we again infer from stationarity:
\begin{align}
 \nu 2^{2j} \E\left( (\bar{u}_j^\nu)^2 \right) + 2^{cj} \E \left(  (\bar{u}_j^\nu)^2 \bar{u}_{j+1}^\nu \right) 
      = 2^{c(j-1)} \E \left( (\bar{u}^\nu_{j-1})^2 \bar{u}_j^\nu\right) + \frac{\sigma^2}{2} \delta_{j - 0},
   \label{eq:stat:bal:1}
\end{align}
for each $j \geq 0$.
Summing \eqref{eq:stat:bal:1} from $j = 0, \ldots, N$ we observe that
\begin{align}
  \nu \sum_{j = 0}^N 2^{2j}  \E\left( (\bar{u}_j^\nu)^2 \right) +  2^{cN} \E \left(  (\bar{u}_N^\nu)^2 \bar{u}_{N+1}^\nu \right)  = \frac{\sigma^2}{2}.
  \label{eq:partial:sum:energy}
\end{align}
In particular we infer that
\begin{align}
   \E \left(  (\bar{u}_N^\nu)^2 \bar{u}_{N+1}^\nu \right) \leq \sigma^2 2^{-cN -1} .
  \label{eq:NL:bnd:1}
\end{align}
We can also deduce from \eqref{eq:partial:sum:energy} and the fact that $\bar{u}^\nu \in H^{+}$
that $\E | \bar u^\nu|_{H^1}^2 \leq \sigma^2/(2 \nu) < \infty$  and thus that 
$\lim_{j \to \infty} 2^{2j} \E |\bar u^\nu_j|^2 =  0$.   This implies with $c/3 \leq 1$ that
\begin{align}
   \nu \sum_{j = 0}^\infty 2^{2j}  \E\left( (\bar{u}_j^\nu)^2 \right) =\nu \E |\bar{u}^\nu|^2_{H^1} \leq \frac{\sigma^2}{2}.
   \label{eq:stationary:H1:balance}
\end{align}
Rearranging in \eqref{eq:stat:bal:0} and using \eqref{eq:NL:bnd:1}
\begin{align}
     \E \left( (\bar{u}^\nu_{j-1})^2 \right)  =& \nu 2^{c} 2^{(2- c)j} \E\left( \bar{u}_j^\nu \right) + 2^{c} \E \left(  \bar{u}_j^\nu \bar{u}_{j+1}^\nu \right) 
     \notag\\
     	&\leq \nu 2^{c} 2^{(2- c)j} \E\left( \bar{u}_j^\nu \right) + 
	2^{c}  \left(\E \left(  (\bar{u}_j^\nu)^2 \bar{u}_{j+1}^\nu \right) \right)^{1/2} 
		\left(\E \left(  ( \bar{u}_{j+1}^\nu)^2 \right) \right)^{1/4}
		     \notag\\
	&\leq   \frac{1}{32} \E \left(  ( \bar{u}_{j+1}^\nu)^2 \right)+ \nu 2^{c} 2^{(2- c)j} \E\left( \bar{u}_j^\nu \right) +
	C \sigma^{4/3} 2^{-2 cj/3}
	\label{eq:strange:cascade:bbd}
\end{align}
Note that the second inequality in this computation was justified by the fact that $\bar{u}^\nu \in H_{+}$. 
Multiplying \eqref{eq:strange:cascade:bbd}  by $2^{(c-1)j}$ and taking the supremum for $1\leq j \leq N+1$, we arrive at
\begin{align*}
&(2^{c-1} - 2^{-c-6}) \sup_{0 \leq j \leq N} 2^{(c-1)j} \E( (\bar u_j^\nu)^2) \notag\\
&\leq \nu 2^c \sup_{0\leq j \leq N+1}   \left( 2^{2j} \E( (\bar u_j^\nu)^2) \right)^{1/2}+ C \sigma^{4/3} \sup_{0\leq j \leq N+1} 2^{(c-1-2c/3)j} + 2^{-c-6} \sup_{N+1 \leq j \leq N+2}   2^{(c-1)j} \E( (\bar u_{j} ^\nu)^2)  \notag\\
&\leq C \nu^{1/2} \left( \nu \E | \bar u^\nu|_{H^1}^2 \right)^{1/2}  + C \sigma^{4/3} \sup_{0\leq j \leq N+1} 2^{(c-1-2c/3)j} + C  2^{(c-3)N}  \left(  \sup_{N+1\leq j \leq N+2} 2^{2j} \E ( (\bar u^\nu_j)^2) \right).
\end{align*}
For $1\leq c \leq 3$ we have $c-1 \leq 2c /3 $ and thus arrive at
\begin{align}
 \sup_{0 \leq j \leq N} 2^{(c-1)j} \E( (\bar u_j^\nu)^2) \leq C \nu^{1/2} \sigma + C \sigma^{4/3} + C  \left(  \sup_{N+1\leq j \leq N+2} 2^{2j} \E ( (\bar u^\nu)^2) \right).
 \label{eq:perrier}
\end{align}
By \eqref{eq:stationary:H1:balance} we have that $\lim_{N\to \infty} 2^{2N} \E( (\bar u_{N}^\nu)^2) = 0$, and upon passing $N\to \infty$ in \eqref{eq:perrier} we obtain
\begin{align}
   \sup_{j \geq 0} 2^{(c-1) j} \E( (\bar u_j^\nu)^2) \leq C \nu^{1/2} \sigma + C \sigma^{4/3} 
\end{align}
which proves \eqref{eq:nu:unif:BDD}.
Now, for $-1 \leq a < (c-1)/2$, the above estimate implies
\begin{align}
\sum_{j=0}^N 2^{2aj} \E ( (\bar u^\nu_j)^2) \leq C ( \nu^{1/2} \sigma + \sigma^{4/3} ) \sum_{j=0}^N 2^{(2a-c+1)j} 
\end{align} 
which proves \eqref{eq:nu:unif:Ha} upon passing $N \to \infty$.

\subsection{Convergence to the Inviscid Model}

Fix any $c\in [1,3]$ and let $\{\bar u^{\nu}\}_{\nu>0}$ be a family of statistically stationary Martingale solutions of \eqref{eq:shell:model:0}--\eqref{eq:shell:model:j}
satisfying \eqref{eq:SSS:positivity}--\eqref{eq:SSS:exp:moment}.
We obtain from the estimates in the previous section the $\nu$-independent bound \eqref{eq:nu:unif:Ha}. 
Since we wish to consider the entire range $c\in[1,3]$, we henceforth fix $a = -1/2$ in \eqref{eq:nu:unif:Ha}.

Fix any $T >0$ and consider the measures
\begin{align*}
  \mu^{\nu}_{E} = \Prb(\bar{u}^{\nu} \in A) \quad A \in \mathcal{B}(C([0,T]; H^{-5})).
\end{align*}
To obtain sufficient compactness to pass to a limit we would like to show that
\begin{align}
  \bar u^{\nu} \mbox{ is uniformly bounded in } L^2(\Omega; L^\infty(0,T;H^{-1/2})).
   \label{eq:uniform:bnd:3}
\end{align}
For this we borrow a trick from \cite{BarbatoMorandinRomito11}.  Working from
\eqref{eq:shell:model:0}--\eqref{eq:shell:model:j} and using that
$\bar u^{\nu} \in H^+$ we infer
\begin{align*}
&d (\bar u^{\nu}_0)^2 + 2 \nu (\bar u^{\nu}_0)^2dt = -(\bar u^{\nu}_0)^2 \bar u^{\nu}_1dt + \sigma^2 dt + 2\sigma \bar u^{\nu}_0 dW,\\
&\frac{d}{dt} \frac{1}{2^j} (\bar u^{\nu}_j)^2 + 2 \nu 2^j (\bar u^{\nu}_j)^2 =- 2\cdot 2^{(c-1)j} (\bar u^{\nu}_j)^2 \bar u^{\nu}_{j+1} + 2^{(c-1)(j-1)} (\bar u^{\nu}_{j-1})^2 \bar u^{\nu}_j \\
	&\qquad \qquad \qquad \qquad \qquad \leq  - 2^{(c-1)j} (\bar u^{\nu}_j)^2 \bar u^{\nu}_{j+1} + 2^{(c-1)(j-1)} (\bar u^{\nu}_{j-1})^2 \bar u^{\nu}_j.
\end{align*}
Summing over $j = 0, \ldots, N$ we obtain:
\begin{align*}
      \sum_{j=0}^N \frac{1}{2^j} (\bar u^{\nu}_j)^2(t) \leq |\bar u^{\nu}(0)|^2_{H^{-1/2}}  + t\sigma^2 + 2 \int_0^t\sigma \bar u^{\nu}_0 dW.
\end{align*}
With Doob's inequality, we now conclude \eqref{eq:uniform:bnd:3}.

In view of the compact embeddings 
\begin{align*}
 L^2([0,T]; H^{-1/2}) \cap W^{1/4,2}([0,T]; H^{-4}) &\subset L^2([0,T]; H^{-1}), \\
 W^{1/4,8}([0,T]; H^{-4}) + W^{1,2}([0,T]; H^{-4}) &\subset C([0,T]; H^{-5}), 
\end{align*}
and using the estimate
\begin{align}
&\Prb\left( \left|\int_0^\cdot (\nu A \bar u^{\nu} + B(\bar u^{\nu})) dt \right|_{W^{1,2}([0,T]; H^{-4})}^2 \geq \frac R 8\right) \notag\\ 
&\qquad \leq \Prb\left(C \sup_{t \in [0,T]}(|\bar u^{\nu}|^2_{H^{-1/2}} + 1) \geq \sqrt{R}\right)   \leq \frac{C}{\sqrt{R}} \E \left( \sup_{t \in [0,T]}|\bar u^{\nu}|^2_{H^{-1/2}} + 1 \right).
     \label{eq:B:R:1:bnd:2}
\end{align}
along with $\Prb( |\sigma W|_{W^{1/4,8}([0,T]; H^{-4})} \geq R)  \leq \frac{C}{R}$ and \eqref{eq:uniform:bnd:3}
we one may deduce that
\begin{align*}
\{\mu^\nu_E\}_{\nu > 0} \mbox{ is tight on } L^2([0,T]; H^{-1}) \cap C([0,T]; H^{-5}).
\end{align*}
See \cite{DebusscheGlattHoltzTemam1} for further details.
We can infer with the Skorokhod embedding theorem as in \cite{Bensoussan1995}  that there exists a probability 
space $(\tilde{\Omega}, \tilde{\mathcal{F}}, \tilde{\Prb})$ and sequence
of solutions stationary martingale solutions $(\tilde{u}^{\nu}, \mathcal{S}_\nu)$ with $\mathcal{S}_\nu = (\tilde{\Omega}, \tilde{\mathcal{F}}, \tilde{\Prb}, \tilde{\mathcal{F}}_t^\nu, \tilde{W}^\nu )$ such that
$\tilde{u}^{\nu} \to \bar{u}  \textrm{ almost surely in } L^2([0,T]; H^{-1}) \cap C([0,T]; H^{-5})$
and $\tilde{W}^\nu \to \overline{W} \textrm{ almost surely in } C([0,T])$.

These convergences are sufficient to show that limiting process $(\bar{u}, \tilde{S})$ is a stationary martingale solutions of the inviscid shell model
\begin{align}
d \bar u_j + (2^{cj} \bar u_j \bar u_{j+1} - 2^{c(j-1)} (\bar u_{j-1})^2 ) dt = \sigma \delta_{j,0} d \overline{W}
\label{eq:inviscid:shell}
\end{align}
with the convention $\bar u_{-1}=0$. Moreover, we infer from $\bar{u}^\nu$ that
\[
\bar u (t) \in H_+ \qquad \mbox{and} \qquad \E | \bar u |_{H^{-1/2}}^2 \leq C.
\]
In fact, a simple argument shows that the uniform in $\nu$ bound  \eqref{eq:nu:unif:BDD} is carried to the limiting stationary solutions $\bar u$, namely we have
\begin{align}
 \sup_{j \geq 0} 2^{(c-1)j}\E ( \bar u_j^2) < \infty.
\label{eq:inviscid:sup:j:bnd}
\end{align}
To see this, fix any $R > 0$.  Observe that by \eqref{eq:nu:unif:BDD} there exists $C < \infty$, independent of $\nu$ and $j$ and $R$, such that 
\[
2^{j(c-1)} \E \left( (\bar{u}_j^\nu)^2  \wedge R \right) \leq  C.
\]
From the Skhorokhod embedding we have $\bar u^\nu_j \to \bar{u}_j$ a.s. for each $j$ as $\nu \to 0$, 
and therefore
\[
2^{j(c-1)} \E \left( (\bar{u}_j)^2  \wedge R \right) \leq  C
\]
via dominated convergence. The monotone convergence theorem and the fact that $\E ( \bar{u}_j^2) < \infty$ for any $j$ proves \eqref{eq:inviscid:sup:j:bnd}, upon sending $R\to \infty$.
Similarly arguing from uniform in $\nu$ bound \eqref{eq:NL:bnd:1} we obtain that
\begin{align}
   \E \left( 2^{cj}  \bar{u}_j^2 \bar{u}_{j+1} \right) \leq \frac{\sigma^2}{2},
   \label{eq:i:did:my:HW:properly:prof:vicol}
\end{align}
which holds for every $j \geq 0$.

\subsection{Enhanced Moment Bounds for the Inviscid Model}

In this section we establish improved regularity, \eqref{eq:best:quadratic:moments}, for the stationary solutions
$\bar{u}$ of \eqref{eq:shell:model:0:inviscid}--\eqref{eq:shell:model:j:inviscid}.

Fix $\eta > 0$ to be determined later. 
For $j\geq 1$, since $\bar u \in H_+$,  upon multiplying \eqref{eq:inviscid:shell} by $1/(\bar u_j + \eta)$ we obtain,
\begin{align}
 2^{c(j-1)} \E \left( \bar u_{j-1}^2 (\bar u_j+ \eta)^{-1}\right)
 = 2^{cj} \E \left( \bar u_{j+1} \bar{u}_j(\bar{u}_j+ \eta)^{-1} \right)  \leq 2^{cj} \E \left( \bar u_{j+1}  \right).
 \label{eq:log:moment}
\end{align}
Now, for $j\geq 2$, since $\bar u_{j-1} \geq 0$ we may use the Cauchy-Schwartz inequality in the above identity.  With \eqref{eq:i:did:my:HW:properly:prof:vicol}
and \eqref{eq:inviscid:sup:j:bnd} to obtain
\begin{align}
\E \left( \bar u_{j-1}^2 \right) 
&\leq \left( \E \left( \bar u_{j-1}^2 (\bar u_{j} + \eta) \right) \right)^{1/2} \left( \E \left(  \bar u_{j-1}^2 (\bar u_{j} + \eta)^{-1} \right) \right)^{1/2} \notag\\
&\leq C( \sigma 2^{-cj/2} + [ \eta  \E( \bar{u}_{j-1}^2)]^{1/2})  \left( \E \left( \bar  u_{j+1} \right) \right)^{1/2} \notag\\
&\leq C_0 \sigma 2^{-cj/2}  \left( \E \left( \bar u_{j+1}^2 \right) \right)^{1/4},
\label{eq:inviscid:regularity:bootstrap:key}
\end{align}
where we obtain the last inequality by setting $\eta = \sigma^2 2^{-j}$.  Note that the constant $C_0$ is independent of $j$ and 
$\sigma$.

Working from \eqref{eq:inviscid:regularity:bootstrap:key} we may now apply the following iterative argument.  Let $b\geq 0$, and assume we know that
\begin{align}
   \sup_{j \geq 0} 2^{jb} \E \left( \bar u_{j}^2 \right) \leq C_b < \infty.
\label{eq:inviscid:Hb:bound}
\end{align}
Let $a\geq 0$. Using \eqref{eq:inviscid:regularity:bootstrap:key} and \eqref{eq:inviscid:Hb:bound} we conclude
\begin{align*}
2^{ja} \E \left( \bar u_j^2 \right)
\leq C_0 \sigma 2^{(a-c/2 - b/4)j}  \left( 2^{b(j+2)} \E \left( \bar u_{j+2}^2 \right) \right)^{1/4} 
\leq C_0 \sigma  2^{(a-c/2 - b/4)j}  C_b^{1/4},
\end{align*}
and therefore, if $a \leq c/2 + b/4$,
we arrive at
\begin{align}
   \sup_{j \geq 0} 2^{ja} \E \left( \bar u_j^2 \right) \leq C_a=: C_0 \sigma C_b^{1/4}.
\label{eq:inviscid:Hb:bound:iterate}
\end{align}
When $b < 2c/3$ in \eqref{eq:inviscid:Hb:bound:iterate} we have gained decay with respect to $j$ in comparison to   \eqref{eq:inviscid:Hb:bound}. This represents an induction step. The base step of the induction argument is given by \eqref{eq:inviscid:sup:j:bnd} above, for $b=c-1$. To conclude, we define 
\[ 
a_1 = c-1 \quad \mbox{and} \quad a_{k+1} = \frac{c}{2} + \frac{a_k}{4},
\]
let $C_1 > 0$ be the constant for which \eqref{eq:inviscid:sup:j:bnd} holds, and define the iteration
\[
C_{k+1} = C_0 \sigma C_k^{1/4}
\]
where $C_0$ is fixed and independent of $\sigma$. By induction, it follows by \eqref{eq:inviscid:Hb:bound} and \eqref{eq:inviscid:Hb:bound:iterate} that
\begin{align}
 \sup_{j \geq 0} 2^{a_k j} \E (\bar u_j^2) \leq C_k
\label{eq:inductive:inviscid:bound}
\end{align}
for all $k\geq 1$. But note that 
\begin{align*}
a_{k+1} = (c-1) 4^{-k} + \frac{c}{2}  \sum_{j=0}^{k-1} 4^{-j}    = \frac{2c}{3} - \frac{3-c}{3 \cdot 4^{k}} \to \frac{2c}{3} \quad \mbox{as} \quad k \to \infty.
\end{align*}
Moreover, we have that 
\begin{align*}
C_{k+1} =C_1^{4^{-k}} (C_0 \sigma)^{\sum_{j=0}^{k-1} 4^{-j} }  \to (C_0 \sigma)^{4/3} \quad \mbox{as} \quad k \to \infty.
\end{align*} 
Thus, passing $k\to \infty$ in \eqref{eq:inductive:inviscid:bound} we arrive at the desired estimate \eqref{eq:best:quadratic:moments}.

\subsection{Anomalous/Turbulent Dissipation}
\label{sec:anomalous:disp}

We finally establish the claims concerning turbulent dissipation stated in item (ii) of Theorem~\ref{thm:inviscid:limits}.
Observe that, for any solution of \eqref{eq:shell:model:0:inviscid}--\eqref{eq:shell:model:j:inviscid}, we infer from the It\={o} lemma 
that
\begin{align}
\frac{d}{dt} \E(|P_N u|^2) = \sigma^2   - 2   \E(\Pi_N(u))
\label{eq:energy:bal:galerkin}
\end{align}
holds for each $N$.
Given any stationary solutions $\bar u$ of  \eqref{eq:shell:model:0:inviscid}--\eqref{eq:shell:model:j:inviscid}
we immediately infer \eqref{eq:flux:bound} from \eqref{eq:energy:bal:galerkin} and stationarity.
We see moreover that $\bar{u}$ satisfies the low regularity bound \eqref{eq:Onsager:moment} since
otherwise
\begin{align}
\lim_{N\to \infty} \E(\Pi_N(\bar{u})) ds = \lim_{N\to \infty}  \E(2^{cN} u_N^2 u_{N+1}) ds = 0,
\label{eq:vanishing:flux:smooth}
\end{align}
in contradiction to \eqref{eq:flux:bound}.
This shows that stationary solutions cannot be smooth and must exhibit anomalous/turbulent dissipation of energy; 
the flux cannot vanish as $N\to \infty$, and the energy balance $\frac{d}{dt} \E (|u|^2) = \sigma^2$ is violated.

\section{Unique Ergodicity and Attraction Properties}
\label{sec:unique:attraction}

In this section we address the question of unique ergodicity and attraction properties for 
the invariant measure associated with \eqref{eq:shell:model:0}--\eqref{eq:shell:model:j} when $\nu > 0$ and
$c$ lies in the range $[1,2)$.  While the existence of an invariant measure follows from the Krylov-Bogolyubov 
averaging procedure (see item (iv) in Proposition~\ref{prop:well:poshness}), the uniqueness of statistically 
steady states is a more delicate issue.  It requires a detailed understanding of the interaction between
the nonlinear and stochastic terms in \eqref{eq:shell:model:0}--\eqref{eq:shell:model:j} as well as 
a number of more involved moment estimates. In Section~\ref{sec:anomal:dissipation} we make use of these results to establish
the anomalous dissipation of energy in the inviscid limit, for $c\in[1,2)$.

Our analysis is carried out in a Markovian framework and makes essential use of the continuous dependence on data (in the topology of $H$), which insofar is 
valid only for  $c \in [1, 2]$.\footnote{See however the generalized framework \cite{Romito11} which builds on \cite{FlandoliRomito08}.}
As described in Section~\ref{sec:smoothing} below, the main step in the proof  is to establish a smoothing
condition  for the Markov semigroup associated to  \eqref{eq:shell:model:0}--\eqref{eq:shell:model:j}, which leads to estimates reminiscent 
of those needed to bound the dimension of the attractor for dissipative dynamical systems~\cite{ConstantinFoiasTemam85,Temam1997}. 
Here the restriction $1\leq c < 2$ plays an important role; the equations are semilinear in this range.

In comparison to previous works on the uniqueness of invariant measures for (semilinear) infinite dimensional systems, 
\cite{HairerMattingly06, HairerMattingly2008,HairerMattingly2011, FoldesGlattHoltzRichardsThomann2013}, a 
new mathematical challenge arrises in verifying an algebraic condition, the so called {\em H\"ormander bracket condition}.
This condition describes the interaction between the nonlinear and stochastic terms and its verification, depending on the structure of the equations, can require
an involved analysis.  It turns out that previous related works, \cite{EMattingly2001,Romito2004,HairerMattingly06, HairerMattingly2011, FoldesGlattHoltzRichardsThomann2013}, 
make significant use of non-local wave number interactions in verifying H\"ormander's condition.  As such the approach taken in these
works can not be repeated here.

After reviewing a few standard preliminaries we introduce the main 
result Theorem~\ref{thm:uniqueness:IM}. In Section~\ref{sec:smoothing} we
briefly recall some generalities which explain the connection between smoothing in the Markovian dynamics,
H\"ormander's condition and question of unique ergodicity.  Section~\ref{sec:Hormander:Cond}
is then devoted to the verification of H\"ormander's condition. The remainder of the 
proof of Theorem~\ref{thm:uniqueness:IM}, while highly nontrivial, is quite similar
to previous works \cite{HairerMattingly06, HairerMattingly2008,HairerMattingly2011, FoldesGlattHoltzRichardsThomann2013}.
Further details are postponed to Appendix~\ref{sec:uniqueness:attraction:proof}.

 \subsection{Markovian setting; Summary of uniqueness and attraction properties of invariant measures}

Before stating Theorem~\ref{thm:uniqueness:IM} we first
recall some generalities and notations for the Markovian framework associated to \eqref{eq:shell:model:0}--\eqref{eq:shell:model:j}. 
For each $\nu >0$ and any $c \in [1,2]$ we define the \emph{Markov transition function}
\begin{align*}
  P_t(\IC, A) = \Prb(u(t, \IC) \in A), \quad \IC \in H, A \in \mathcal{B}(H),
\end{align*}
where $u(t, \IC)$ is the unique pathwise solution of  \eqref{eq:shell:model:0}--\eqref{eq:shell:model:j} 
and $\mathcal{B}(H)$ are the Borel subsets of $H$.
We then we define the Markov semigroup
\begin{align}
  P_t \phi(\IC) = \E \phi(u(t, \IC)) = \int_H  \phi(u) P_t(\IC, du),
  \label{eq:msg}
\end{align}
for any $\phi \in M_b(H)$.  Here $M_b(H)$ denotes the collection of real valued, measurable and bounded
function on $H$.
We take $P_t^*$ (which is the dual of $P_t$) according to 
\begin{align*}
  P_t^* \mu(A) = \int_H P_t(\IC, A) d\mu(\IC)
\end{align*}
for elements $\mu \in Pr(H)$, the collection of Borealian probability measures on $H$.
An element $\mu \in Pr(H)$ is an \emph{invariant measure} of the Markovian semigroup if 
it is a fixed point of $P_t^*$ for every $t  \geq 0$.  Such elements represent statistically steady
states of \eqref{eq:shell:model:0}--\eqref{eq:shell:model:j}.

Take $C_b(H)$ to be the collection of real valued continuous bounded functions mapping from $H$. 
Recall that $P_t$ is said to be \emph{Feller} if $P_t: C_b(H) \to C_b(H)$
for every $t \geq 0$.   This property is needed for all that follows and 
indeed some form of the Feller property is required even to prove the existence of an invariant measure
of \eqref{eq:shell:model:0}--\eqref{eq:shell:model:j}. With this in mind, we now specialize to case $c \in [1,2]$.  In this situation observe that if $\IC^{n} \to \IC$ in $H$ then, 
in view of \eqref{eq:con:dep:est}, $u(t, \IC^{n}) \to u(t, \IC)$ a.s. in $H$.   It follows from the dominated 
convergence theorem that $P_t \phi(\IC^{n}) \to P_t \phi(\IC)$ which establishes that $P_t$ is Feller
when $c \in [1,2]$.

Beyond $\mathcal{M}_b(H)$ and $C_b(H)$ we will make use of several further classes
of test functions on $H$.  Define
\begin{align*}
   \| \phi \|_\gamma  := \sup_{u \in H} \exp(- \gamma |u|^2) \left( | \phi(u)| + | \nabla \phi(u)|^2  \right)
\end{align*}
and take
\begin{align}
  \mathcal{B}_\gamma :=\{ \phi \in C^1(H): \| \phi \|_\gamma < \infty\},
  \quad \mathcal{G} := \{\phi \in C^1(H): \| \phi \|_\gamma < \infty, \textrm{ for each } \gamma > 0\}
  \label{eq:exp:grow:space}
\end{align}
We also consider the classes acting on higher regularity space with at most polynomial growth at infinity namely
\begin{align*}
   \mathcal{P}_{m, p} := \left\{ \phi \in C^1(H^m):  \sup_{u \in H^m}  \frac{ |\phi(u)| + |\nabla \phi(u)| }{1 + | u |_{H^m}^p}  < \infty  \right\}
\end{align*}
for any $m \geq 0$ and any $p \geq 2$.

With these preliminaries in hand we state main results concerning the uniqueness and attraction
properties of invariant measures for $P_t$ as follows:
\begin{theorem}
\label{thm:uniqueness:IM}
Suppose that $c \in [1,2)$, $\nu > 0$ and consider solutions $u(t, \IC)$ of the stochastic dyadic shell model \eqref{eq:shell:model:0}--\eqref{eq:shell:model:j}
corresponding to any initial condition $\IC \in H$.  Then there exists a unique invariant measure $\mu_\nu$ of the corresponding 
Markov semigroup which is ergodic.  More precisely, for any $t > 0$, $P_t$ is ergodic with respect the probability space 
$(H, \mathcal{B}, \mu_\nu)$ and this implies that, for any $\phi \in L^2(H; \mu_\nu)$,
\begin{align}
  \frac{1}{T} \E \int_0^T  \phi(u(t,\IC) dt \to \int_H \phi(u) d \mu_\nu (u),
  \label{eq:ergodicity:time:average:form}
\end{align}
for $\mu_\nu$ almost every $\IC$.
Additionally, the invariant measures $\mu_\nu$ obey the attraction properties
\begin{itemize}
\item[(i)] (Mixing) For any $\eta > 0$ there exists positive constants $\gamma_1, \gamma_2 > 0$ (depending on $\nu, c, \eta$)
such that
\begin{align}
   \left|  \E \phi( u(t, \IC)) - \int_H \phi(u) d \mu_\nu(u) \right| \leq C \exp( - \gamma_1 t + \eta | \IC |^2)  \| \phi \|_{\gamma_2}
     \label{eq:mixing:wash:norm}
\end{align}
holds every $\phi \in \mathcal{B}_{\gamma_2}$ and any $\IC \in H$.  Moreover for any $m \geq 0$, $p \geq 2$
and $\phi \in \mathcal{P}_{m,p}$
\begin{align}
  \lim_{T \to \infty} \E \phi(u(T,\IC)) = \int \phi(u) d\mu_\nu(u).
  \label{eq:more:mixing}
\end{align}
\item[(ii)] (Strong law of large numbers) For every $\phi \in \mathcal{G}$ and any $\IC \in H$,
\begin{align}
  \frac{1}{T} \int_0^T \phi(u(t,\IC)) dt \rightarrow \int_H \phi(u) d \mu_\nu(u) \quad \textrm{ almost surely}.
  \label{eq:SLLN}
\end{align}
\item[(iii)] (Central limit theorem) For each $\phi \in C_b^1(H)$, $\IC \in H$ define
\begin{align*}
  m_\phi := \int_H \phi(u) d \mu_\nu(u), \quad v_\phi := \lim_{T \to \infty} \frac{1}{T} \E \left(\int_0^T (\phi(u(t,\IC) - m_\phi) dt \right)^2,
\end{align*}
and let $F_{\phi}$ be the distribution function of a normal random variable with mean $0$ and variance $v_\phi$.
Then, for any $x \in \RR$
\begin{align}
  \lim_{T \to \infty} \Prb \left( \frac{1}{\sqrt{T}} \int_0^T (\phi(U(t, U_0)) - m_\phi) dt    < x \right) = F_{\phi}(x).
  \label{eq:CLT}
\end{align}
In other words  $\frac{1}{\sqrt{T}} \int_0^T (\phi(U(t, U_0)) - m_\phi)  dt$ converges in distribution to normal random variable
with mean $0$ and variance $v_\phi$.
\end{itemize} 
\end{theorem}

\subsection{Smoothing of the Markovian Semigroup in Infinite Dimensions}
\label{sec:smoothing}
We turn next to describe the key ingredients that we use to prove the Theorem~\ref{thm:uniqueness:IM}.
We follow a strategy going back to  
Doob \cite{Doob1948} and Khasminskii \cite{Khasminskii1960}.   
These results identify that uniqueness
and attraction properties similar to Theorem~\ref{thm:uniqueness:IM} hold 
when $P_t$ is \emph{strong Feller} meaning that
$P_t$ maps bounded measurable functions to continuous functions and \emph{irreducible}
which says that from any starting point in the phase space there is a non-zero probability of 
ending up in any other part of the phase space after a finite time.  

Both the strong Feller property and irreducibility condition are too stringent for infinite dimensional systems where the 
stochastic forcing acts directly in only a few directions in phase space, as is the case with our model 
\eqref{eq:shell:model:0}--\eqref{eq:shell:model:j}. Inspired by the insights of recent works \cite{HairerMattingly06, HairerMattingly2008,HairerMattingly2011} 
Theorem~\ref{thm:uniqueness:IM} can be shown to follow from the following two weaker properties.  
The first condition, replacing classical irreducibility, requires that only one point is
universally reachable in phase space.
\begin{proposition}
\label{thm:irr:simple}
For any $\epsilon > 0$, $R > 0$ there exist a time $t^* = t^*(\epsilon,R)$ such that
\begin{align}
  \sup_{\IC \in H,  | \IC | \leq R}  \Prb( | u(t, \IC) | < \epsilon) > 0
\end{align}
for every $t > t^*$.
\end{proposition}
The second estimate immediately implies a form of infinite time smoothing \`a la the \emph{asymptotic strong Feller condition}
introduced in \cite{HairerMattingly06}.
\begin{proposition}
\label{thm:grad:est:MSG}
For any $\gamma, \eta > 0$
\begin{align}
  \| \nabla P_t \phi(\IC) \| \leq C \exp( \gamma |\IC | ) \left( \sqrt{ P_t (|\phi|^2)(\IC) } + \exp(- \eta t) \sqrt{ P_t( \| \nabla \phi \|^2)(\IC)}  \right)
  \label{eq:main:grad:MSG:est}
\end{align}
for every $\phi \in C^1_b(H)$, $\IC \in H$ where the constant $C = C(\gamma, \eta)$ is independent of $t$ and $\phi$ and $\IC$.
\end{proposition}

Proposition~\ref{thm:irr:simple} is an expression of the triviality of the long term dynamics of the unforced version of 
\eqref{eq:shell:model:0}--\eqref{eq:shell:model:j}.
This may be demonstrated precisely as in \cite{EMattingly2001,ConstantinGlattHoltzVicol2013}.
Thus, the main step to establish Theorem~\ref{thm:uniqueness:IM} is to prove the gradient 
estimate Proposition~\ref{thm:grad:est:MSG} on the Markovian semigroup $\{P_t\}_{t \geq 0}$ associated to \eqref{eq:shell:model:0}--\eqref{eq:shell:model:j}
via \eqref{eq:msg}.

The estimate \eqref{eq:main:grad:MSG:est} establishes a form of smoothing for $P_t$.  
Observe that $\psi(\IC) = P_t \phi(\IC)$ formally solves the Kolmogorov backward equation
\begin{align}
  \partial_{t} \psi(\IC,t) = \frac{\sigma^2}{2} \partial_{0}^2 \psi(\IC,t) - \sum_j \langle \nu A(\IC) + B(\IC), e_j \rangle \partial_j \psi(\IC, t);\quad  \psi(0,\IC) = \phi(\IC).
\end{align}
which is a degenerately parabolic system.  Following the analysis in \cite{HairerMattingly06, HairerMattingly2011} which generalizes
the classical hypo-elliptic theory \cite{Hormander1967} we will therefore need to establish a form of the H\"ormander bracket condition 
in order to expect the (asymptotic) smoothing required by \eqref{eq:main:grad:MSG:est}.  

In the next section we recall in our notations and framework the form of this condition introduced in \cite{HairerMattingly06, HairerMattingly2011}.   
The verification of this condition is the main mathematical novelty in the proof of Proposition~\ref{thm:grad:est:MSG}. 
Having established this condition the rest of the analysis leading to \eqref{eq:main:grad:MSG:est} and hence Theorem~\ref{thm:uniqueness:IM}
follows closely previous works \cite{HairerMattingly06, HairerMattingly2008, HairerMattingly2011,FoldesGlattHoltzRichardsThomann2013}.  
We therefore postpone the rest  of the proof of Theorem~\ref{thm:uniqueness:IM} for the Appendix~\ref{sec:uniqueness:attraction:proof}.

\subsection{The H\"ormander  Condition}
\label{sec:Hormander:Cond}

We introduce the infinite dimensional version of the H\"ormander bracket condition as follows.
If $G_1$ and $G_2$ are Frechet differentiable maps on $H$ we define the \emph{Lie bracket} of $G_1$ and $G_2$ according to
\begin{align}
   [G_1, G_2](u) = \nabla G_2(u) G_1(u) - \nabla G_1(u)  G_2(u).
   \label{eq:lie:brak:def}
\end{align}
Take $e_j = (\delta_{i - j})_{i \geq 0}$  and let
\begin{align*}
   F(u) =  \nu A u + B(u,u) 
\end{align*}
where we have symmetrized the bilinear form $B$ so that 
\begin{align}
  B(u,v)_j = 2^{cj-1} u_{j+1} v_{j} + 2^{cj-1}v_{j+1} u_{j}  -2^{c(j-1)} v_{j-1} u_{j-1}.
  \label{eq:sim:B:NLT}
\end{align}
In our context the H\"ormander condition states that we can approximate the phase space $H$ with a sequence of 
allowable Lie brackets staring from $e_0$.  We may then proceed to fill $H$ by then taking successive brackets involving either
$F$ or $e_0$ with previously obtained vector fields.   More precisely we make the following definitions
\begin{definition}[H\"ormander's condition]\label{def:Hormander:bracket}
Let $\AddSet_0 := \mbox{span} \left\{ e_0 \right\}$ and iteratively define
\begin{align}
\AddSet_m :=  \mbox{span} \left\{ [G(u), e_0], [G(u), F(u)], G(u) : G \in \AddSet_{m-1}\right\}.
 \label{eq:add:set:def}
\end{align}
We say elements $E \in  \cup_m \AddSet_m$ are \emph{admissible vector fields} which have been produced by an \emph{admissible sequence of Lie
Brackets}.  The system \eqref{eq:shell:model:0}--\eqref{eq:shell:model:j} is said to satisfy the \emph{H\"ormander bracket condition}
if
\begin{align}
  \textrm{for every N, there exists $m = m(N)$ such that }   \AddSet_m \supset H_N.
  \label{eq:easiest:inf:dim:hormander}
\end{align}
\end{definition}

Compared to previous analogous results which have been obtained 
for the Navier-Stokes nonlinearity in \cite{EMattingly2001, Romito2004, HairerMattingly06, HairerMattingly2011}
it would seem at first glance that the analysis of nonlinear structure in $B$, cf. \eqref{prop:hor:cond}, leading to \eqref{eq:easiest:inf:dim:hormander}
would be easier to address.  Indeed, observe that
\begin{align}\label{eq:b:struct}
   B(e_j,e_k) = 
   \begin{cases}
   - 2^{cj}  e_{j+1},& \textrm{ when } k = j\\
      2^{cj-1}  e_{j},& \textrm{ when } k= j+1\\
     0,& \textrm{ when } |k- j| \geq 2.\\
   \end{cases}
\end{align}
Actually, it is this nearest neighbor only interaction that 
leads to new difficulties in comparison to these previous works.
Naively we may fill the phase space by iteratively taking Lie brackets of the form
\begin{align*}
  [[F(u), e_k],e_k] = 2 B(e_k, e_k) = -2^{ck + 1} e_{k+1}
\end{align*}
Unfortunately, it is not clear that such brackets are admissible 
in the sense of Definition~\ref{def:Hormander:bracket} and a more
careful analysis of the interaction between $F$ and $e_0$ is needed
to ensure that \eqref{eq:easiest:inf:dim:hormander} is satisfied.\footnote{For comparison
in \cite{EMattingly2001} brackets of the form $[[B(u),e_j], e_0]$ are used to generate the
phase space.  As such this work actually makes use significant use of the 
\emph{long range interactions} (in wave space) present in the nonlinear terms.} 

To overcome this complication, we consider the polynomials of the form
\begin{align}
  S_0(v_1, v_2) &= B(v_1, v_2) \notag\\ 
  S_1(v_1, v_2, v_3, v_4) &= S_0(B(v_1, v_2), B(v_3, v_4)) \notag\\ 
  S_2(v_1, \ldots , v_{8})& = S_1(B(v_1, v_2), B(v_3,v_4), B(v_5,v_6), B(v_7, v_8)) \notag\\ 
  &\, \, \,  \vdots \notag\\
  S_m(v_1, \ldots, v_{2^{m+1}}) &= S_{m-1}(B(v_1, v_2), \ldots, B(v_{2^{m+1} -1}, v_{2^{m+1}})).
  \label{eq:the:last:bracket}
\end{align}
By bracketing $[F, e_{m+1}]$ repeatedly against $F$, $2^{m}$ times we will show that
the resulting admissible vector fields have the form
\begin{align}
   \mathfrak{S}_{m+1}(u) := [\ldots [[F, e_{m+1}], F], \ldots, F ](u) = C_m B(e_{m+1}, S_{m}(u, \ldots, u)) + \mathcal{E}_m(u),
  \label{eq:iteration:atmpt}
\end{align}
where $C_m \not = 0$ and $\mathcal{E}_m$ has an involved structure.  Bracketing $\mathfrak{S}_{m+1}(u)$ repeatedly against $e_0$
yields further admissible vector fields and as we will see, $S_m(e_0, \ldots, e_0)  \sim e_{m+1}$.  On the other 
hand one we will show that $\mathcal{E}_m(e_0, \dots, e_0) \in \mbox{span}\{e_0, \ldots, e_{m+1}\}$ in order to avoid 
possible cancelations with $C_m B(e_{m+1}, S_{m}(e_0, \ldots, e_0))$ preventing the generation of new directions in $H$
with this strategy.   

With these motivating discussions in mind the rest of the section is devoted to proving:
\begin{theorem}\label{prop:hor:cond}
  The dyadic model \eqref{eq:shell:model:0}--\eqref{eq:shell:model:j} satisfies the \emph{H\"ormander bracket condition}~\eqref{eq:easiest:inf:dim:hormander}.
\end{theorem}

We begin by introducing some further notations.  Let $\mathcal{M}_1 = \{ A^k u: k \geq 0\}$, and take
\begin{align*}
   \mathcal{M}_2 = \{ A^j B(A^l u, A^m u) : j, l, m \geq 0, l \geq m\},
\end{align*}
and for $k \geq 2$ define iteratively:
\begin{align}
  \mathcal{M}_k = \{ \tilde{B}(E(u, \ldots, u), u), \tilde{B}( u,E(u, \ldots, u)), E(&\tilde{B}(u),u, \ldots, u), 
  	                        \ldots, E(u, \ldots, u ,\tilde{B}(u)): \notag\\
	                        \tilde{B} \in \mathcal{M}_2, E \in \mathcal{M}_{k-1}
	                        \}  
	                        \label{eq:kth:order:iterated:poly}                      
\end{align}
Note carefully that $\mathcal{M}_k$ consists of $k$ linear forms.  Moreover,
for any $E \in \mathcal{M}_k$, a simple induction shows that $E$ has the form
 \begin{align}
   E(u) = \tilde{B}(E_1(u), E_2(u)) \quad \textrm{ where } E_1 \in \mathcal{M}_{l_1}, E_2 \in \mathcal{M}_{l_2}, \tilde{B} \in \mathcal{M}_2
   \textrm{ and } l_1 + l_2 = k.
   \label{eq:E:division}
 \end{align}
We also take
$\mathcal{S}_0 = \mathcal{M}_2$
and for $m \geq 1$ define
\begin{align}
 \mathcal{S}_m = \{ \tilde{S}_{m-1}(\tilde{B}^1(u), \ldots , \tilde{B}^{2^m}(u)): \tilde{S}_{m-1} \in \mathcal{S}_{m-1}, \tilde{B}^i \in \mathcal{S}_0 \}
 \label{eq:sim:sets:sm}
\end{align}
Observe that $\mathcal{S}_m \subset \mathcal{M}_{2^m}$ and that $\tilde{S}_m \in \mathcal{S}_m$.
Also note that we can equivalently build
\begin{align}
  \mathcal{S}_m = \{ \tilde{B}(\tilde{S}_{m-1}^1, \tilde{S}_{m-1}^2): \tilde{B} \in \mathcal{S}_0, \tilde{S}^i_{m-1} \in \mathcal{S}_{m-1} \}.
  \label{eq:S:alt:build}
\end{align}

We have the following lemma
\begin{lemma}\label{lem:pain:over:graphs}
  For every $m \geq 0$ and each $\tilde{S}_m \in \mathcal{S}_m$
  \begin{align}
     \tilde{S}_m(e_0) = C_{\tilde{S}_m} e_{m+1}
     \label{eq:E:carefully:built:is:good}
  \end{align}
  where $C_{\tilde{S}_m}$ is a suitable non-zero constant.
  Moreover, for every $k \geq 2$ and every $E \in \mathcal{M}_k$
  such that $E_k \not\in  \mathcal{S}_m$ for some $m$
  \begin{align}
  	E(e_0) =
	 C_{E} e_l \quad \textrm{ for some } l \leq \lceil \log_2( k) \rceil, 
	 \label{eq:E:cannot:be:bad:cone}
  \end{align}
  for a constant $C_E$ depending on $E$ which may be zero.
\end{lemma}
\begin{proof}
The first identity \eqref{eq:E:carefully:built:is:good} follows from \eqref{eq:b:struct}
and \eqref{eq:S:alt:build} with an induction argument on $m$.

The proof of \eqref{eq:E:cannot:be:bad:cone} 
 is an induction on $m \geq 1$ making use of \eqref{eq:b:struct}, \eqref{eq:E:division}, \eqref{eq:E:carefully:built:is:good}.  The inductive hypothesis is that the condition \eqref{eq:E:cannot:be:bad:cone} holds for each $k \leq 2^m$.
The base case follows from \eqref{eq:b:struct} by inspection.  Suppose then
that  \eqref{eq:E:cannot:be:bad:cone} holds for all $k \leq 2^m$ and consider any
$2^m< k \leq 2^{m+1}$ and any $E \in \mathcal{M}_k$ with $E \not \in \mathcal{S}_{m+1}$.  By  
\eqref{eq:E:division}, $E(u) = \tilde{B}(E_1(u), E_2(u))$, $\tilde{B} \in \mathcal{M}_2$ where, without loss of 
generality $E_1 \in \mathcal{M}_{\tilde{k}}$ with $\tilde{k} \leq 2^m$.   Two situations
may arise.  Firstly we may have that $\tilde{k} = 2^m$ and $E_1 \in \mathcal{S}_m$.  In this case
$E_2 \in \mathcal{M}_{k - 2^m}$ and moreover it cannot lie in $\mathcal{S}_m$ (or else we would 
contradict that $E \not \in S_{m+1}$).  We infer, by the inductive
hypothesis, that $E_2(e_0) = c e_j$
for some $j \leq m$ and hence with \eqref{eq:b:struct} conclude $E(e_0) = C' B(e_{m+1}, e_j) = C e_j$ 
(where $C'$,$C$ may be zero).  The second possibility is that $E_1 \not \in \mathcal{S}_m$
in which case, again with the inductive hypothesis $E_1(e_0) = Ce_j$ where $j \leq m$ and 
$E_2(e_0) = Ce_l$ (where $l$ may indeed by greater than $m$).  Combining these two observations
we finally infer $E(e_0) = C'E(e_j, e_l) = Ce_{\tilde{j}}$ where $\tilde{j} \leq m+1$.  This completes the induction 
and hence the proof of Lemma~\ref{lem:pain:over:graphs}.
\end{proof}

With these preliminaries in hand we now show that \eqref{eq:easiest:inf:dim:hormander}
is satisfied as follows. 

\begin{proof}[Proof of Theorem~\ref{prop:hor:cond}]
 Observe that, for any $m \geq 0$,
\begin{align*}
  [F, e_{m+1}] = \nu A e_{m+1} + 2 B(e_{m+1}, u)
\end{align*}
So that bracketing by $[F, e_{m+1}]$ repeatedly against $F$, $2^{m}$ times we obtain
a vector field
$\mathfrak{S}_{m+1}(u)$ of the form \eqref{eq:iteration:atmpt}
where the constant $C_m$ is non-zero, and $\mathcal{E}_m$ is a polynomial which has the form
\begin{align}
  \mathcal{E}_m(u) =& \sum_{k =1}^{2^{m}-1} \sum_{E \in \mathcal{M}_k} C_E B(e_{m+1}, E(u))
      + \sum_{E \in \mathcal{M}_{2^{m}} \setminus \mathcal{S}_m} C_E B(e_{m+1},E(u)) \notag\\
      &+ \sum_{\substack{k_1 + k_2 = 2^{m}+1\\k_1 \geq 2}} \sum_{E_1 \in \mathcal{M}^I_{k_1}, E_2 \in \mathcal{M}_{k_2}} C_{E_1,E_2} B(E_1(e_{m+1},u), E_2(u))
      \notag\\
         &+ \sum_{\substack{k_1 + k_2 \leq 2^{m}\\k_1 \geq 1}} \sum_{\substack{E_1 \in \mathcal{M}^I_{k_1}, E_2 \in \mathcal{M}_{k_2}\\ \tilde{B} \in \mathcal{M}_2}} C_{E_1,E_2, \tilde{B}}\tilde{B}(E_1(e_{m+1},u), E_2(u))
      \label{eq:nasty:poly}
\end{align}
where
\begin{align*}
  \mathcal{M}_k^I := \{ E(v,u): H \times H& \to H: \notag\\
  	                                &E \in \mathcal{M}_k, E(v,u) = E(v, u, \ldots, u), E(u, v, u, \ldots, u), \ldots, E(u, \ldots, u, v)\}.
\end{align*}
With \eqref{eq:E:cannot:be:bad:cone}, \eqref{eq:b:struct} and a careful inspection of \eqref{eq:nasty:poly} we find that
\begin{align}
   \mathcal{E}_m(e_0)  \in \mbox{span}\{ e_0, \ldots, e_{m+1}\}.
   \label{eq:good:span:junk}
\end{align}
Observing that taking Lie brackets of $\mathfrak{S}_{m+1}$ with $e_0$,
$2^m$ times we obtain
\begin{align}
  [\ldots [ \mathfrak{S}_m(u), e_0], \ldots, e_0] = \mathfrak{S}_m(e_0) 
  = \tilde{C}_{m} e_{m+2} + \mathcal{E}_m(e_0),
  \label{eq:spanning:now:everything}
\end{align}
where $\tilde{C}_m$ is a non-zero constant.
Arguing inductively we see that $\tilde{C}_{m} e_{m+2} + \mathcal{E}_m(e_0)$ is produced by an admissible sequence of
Lie brackets. Thus with \eqref{eq:good:span:junk} and \eqref{eq:spanning:now:everything} we see that the H\"ormander
bracket condition of the form given in \eqref{def:Hormander:bracket} is satisfied, completing the proof of Theorem~\ref{prop:hor:cond}. 
\end{proof}

\section{Dissipation Anomaly in the Inviscid Limit}
\label{sec:anomal:dissipation}

In this final section we establish the dissipation anomaly in the inviscid limit. We  prove the following:
\begin{theorem}\label{thm:dissipation:anomaly}
Fix any $c \in [1,2)$ and let $u^\nu(\cdot, \IC)$ be the unique solution of \eqref{eq:shell:model:0}--\eqref{eq:shell:model:j}
for any $\IC \in H$.  Then
\begin{align}
   \lim_{\nu \to 0} \lim_{T \to \infty}  \nu \E |u^\nu(T, \IC)|^2_{H^1} = \frac{\sigma^2}{2}.
   \label{eq:dissipation:anomaly:mixing:ver}
\end{align}
Moreover, for any such $\IC \in H$
\begin{align}
   \lim_{\nu \to 0} \lim_{T\to \infty}  \frac{\nu}{T} \int_0^T |u^\nu(t, \IC)|^2_{H^1} dt = \frac{\sigma^2}{2}.
     \label{eq:dissipation:anomaly:SLLN:ver}
\end{align}
\end{theorem}

\begin{proof}
We immediately infer 
\eqref{eq:dissipation:anomaly:mixing:ver} from \eqref{eq:more:mixing} and energy balance
in \eqref{eq:shell:model:0}--\eqref{eq:shell:model:j}.
Indeed let $\bar{u}^\nu$ be the stationary solution corresponding to $\mu_\nu$.  Then 
$\nu \E |\bar{u}^\nu |^2_{H^1} = \sigma^2/2$ so, making use of \eqref{eq:more:mixing}
we conclude that
\begin{align*}
  \nu \E |u^\nu(T, \IC) |^2_{H^1} \to \nu \int | u |_{H^1}^2 d \mu(u) =  \frac{\sigma^2}{2}
\end{align*}
for any $\IC \in H$.

For the second item, \eqref{eq:dissipation:anomaly:SLLN:ver} take 
\begin{align*}
\psi_N(u) = \nu \sum_{j =0}^{N} 2^{2j} u_{j}^{2} = \nu | P_N u|_{H^1}^2
\end{align*}
and notice that $\psi_N$ is in the set $\mathcal{G}$ is defined in \eqref{eq:exp:grow:space}.  We infer from \eqref{eq:SLLN} that for 
any $\IC \in H$
\begin{align*}
  \liminf_{T \to \infty} \frac{\nu}{T} \int_0^T |u^\nu(t, \IC) |^2_{H^1} dt \geq   \liminf_{T \to \infty} \frac{1}{T} \int_0^T \psi_N(u^\nu(t, \IC)) dt
    = \int \psi_N(u) d\mu(u).
\end{align*}
Now, by the monotone convergence theorem
\begin{align*}
   \lim_{N \to \infty} \int \psi_N(u) d\mu(u) = \lim_{N \to \infty} \nu \E | P_N \bar{u}^\nu |^2_{H^1} = \nu \E | \bar{u}^\nu |^2_{H^1} = \frac{\sigma^2}{2},
\end{align*}
so that 
\begin{align*}
 \liminf_{T \to \infty} \frac{\nu}{T} \int_0^T |u^\nu(t, \IC)|^2_{H^1} \geq  \frac{\sigma^2}{2}.
\end{align*}
For a suitable upper bound observe that, due to the It\={o} Lemma,
\begin{align*}
  \frac{1}{T} \left( |u(T, \IC) |^2  + 2\nu \int_0^T |u(t, \IC)|^2_{H^1} dt \right) = \frac{1}{T} \left( |\IC |^2  + \frac{\sigma^2 T}{2} + 2\sigma \int_0^T u_0 dW \right)
\end{align*}
Thus, the second item \eqref{eq:dissipation:anomaly:SLLN:ver} is proven once we establish that
\begin{align}
 \frac{1}{T}\int_0^T u_0 dW \to 0, a.s.
 \label{eq:stochastic:int:to:zero}
\end{align}
For $\delta \in (0,1)$ and $n$ define
\begin{align*}
   M_n := \int_0^{n \delta}  u_0 dW, \quad  X_k = \int_{\delta(k-1)}^{\delta k} u_0 dW.
\end{align*}
With the It\={o} isometry we have
\begin{align*}
   \E X_k^2 = \E \int_{\delta(k-1)}^{\delta k} u_0^2 ds \leq   \int_{\delta(k-1)}^{\delta k} \E | u(s,\IC) |^2 ds.
\end{align*}
Now, since $\frac{d}{dt} \E |u|^2 + 2\nu \E |u|^2 \leq \sigma^2$, we have that
\begin{align*}
  \E |u(t,\IC)|^2 \leq \exp(-2\nu t) | \IC |^2 + \frac{\sigma^2}{2\nu}.
\end{align*}
With these observations we infer that
\begin{align*}
  \sum_{k =1}^\infty \frac{\E X_k^2}{(\delta k)^2} \leq \frac{| \IC |^2 + \frac{\sigma^2}{2\nu}}{\delta} \sum_{k} k^{-2} < \infty
\end{align*}
By the Martingale SLLN (see for example \cite[Theorem 7.21.1]{KuksinShirikian12}) we infer thus \eqref{eq:stochastic:int:to:zero}, completing the proof of the theorem.\footnote{Actually
this implies the $F(T) = \frac{1}{T}\int_0^T u_0 dW$ goes to zero along any sequence on a dense subset of $[1, \infty)$.  Since 
$F(T)$ is almost surely continuous this implies that this convergence occurs along \emph{any} sequence.}
\end{proof}

\appendix

\section{Gradient Estimates for the Markov Semigroup}
\label{sec:uniqueness:attraction:proof}

In this section we sketch some further details of the proof of Theorem~\ref{thm:uniqueness:IM}.    
The approach closely follows  
the recent works \cite{HairerMattingly06, HairerMattingly2008, HairerMattingly2011, KuksinShirikian12, FoldesGlattHoltzRichardsThomann2013}
modulo the analysis establishing the H\"ormander bracket condition which is carried out in Section~\ref{sec:Hormander:Cond}.   
In sections~\ref{sec:control:prob}--\ref{sec:cost:control} we describe and solve a control problem
which implies Proposition~\ref{thm:grad:est:MSG}.  The solution of this
problem requires a Foias-Prodi type bound for  a linearization of \eqref{eq:shell:model:0}--\eqref{eq:shell:model:j} 
as well as an estimate on the spectrum of
an operator (the Malliavin covariance matrix) associated to this linearization.  We describe how
these bounds are achieved in section~\ref{sec:FoiasProdi:bnds} and \ref{sec:mal:mat:analysis}.
The final section explains how one derives 
Theorem~\ref{thm:uniqueness:IM} from Proposition~\ref{thm:grad:est:MSG}.

\subsection{Smoothing as a control problem}
\label{sec:control:prob}

The first step in the proof of \eqref{eq:main:grad:MSG:est} is to translate this bound into a control problem. 
For this purpose we introduce some linearization operators around \eqref{eq:abstract:shell}.
Fix any $\xi,\IC \in H$, and any $0 \leq s \leq t \leq T$ take $\rho = \JJ_{s,t} \xi$
to be the solution of
\begin{align} 
& \frac{d}{dt} \rho + \nu A \rho + B(u,\rho) + B(\rho,u) = 0, \quad  \rho(s) = \xi, \label{eq:rho:def:1}
\end{align}
where $u = u(t,\IC)\in C(0,T;H) \cap L^2(0,T;H^1) $ obeys \eqref{eq:abstract:shell}.
For $s < t$ and $v \in L^2([s,t])$ we let
\begin{align}
 \AAA_{s,t} v := \sigma \int_s^t \JJ_{r,t} e_0 v(r) dr,
 \label{eq:A:op}
\end{align}
where $e_0 = (1, 0, 0, \ldots ) \in H$.  The processes $\JJ_{0,t} \xi$ and $\AAA_{0,t}v$ represent infinitesimal 
perturbations of $u$ in its initial conditions and driving noise in the directions $\xi$ and $v$ respectively.
Using the Malliavin chain rule and integration by parts formulas (see \cite{Nualart2006}) one obtains that, for any $\xi \in H$ 
and any suitable $v \in L^2(0,t)$\footnote{Here we do not require that $v$ is adapted so that $\int_0^t v dW$ is in general only a Skorokhod integral.  See \cite{Nualart2006}.}
\begin{align*}
   \nabla P_t \phi(\IC) \xi = \E \left(  \phi(u(t,\IC)) \int_0^t v dW \right) + \E \left( \nabla \phi(u(t, \IC)) (\JJ_{0,t}\xi - \AAA_{0,t} v)\right), \quad t \geq 0.
\end{align*}
Notice that $\bar{\rho}(t) = \JJ_{0,t}\xi - \AAA_{0,t} v$ solves $\frac{d}{dt} \bar{\rho} + \nu A \bar{\rho} + B(u,\bar{\rho}) + B(\bar{\rho},u) =-\sigma e_0 v$, where
$\bar{\rho}(0) = \xi$.
With the H\"older inequality we now see that the proof of \eqref{eq:main:grad:MSG:est} reduces to proving:
\begin{proposition}
\label{prop:cont:problem}
For every $\xi \in H$ there exists a corresponding $v = v(\xi) \in L^2([0,\infty))$
such that 
\begin{align}
  \sup_{ \xi \in H, \|\xi\| = 1}  \E ( | \JJ_{0,t} \xi - \AAA_{0,t} v(\xi)|^2) \to 0  \quad \textrm{ as } t \to \infty,
  \label{eq:decay:cnt}
\end{align}
and such that
\begin{align}
 \sup_{t \geq 0}\sup_{ \xi \in H, |\xi | = 1}\E \left(\int_0^t v(\xi) dW\right)^2 < \infty.
 \label{eq:cost:cnt}
\end{align}
\end{proposition}

\subsection{Defining the control}

A suitable choice for the control $v$ can be obtained in terms of the \emph{Malliavin covariance matrix} or \emph{control Grammian}
$\MM_{s,t} = \AAA_{s,t} \AAA^*_{s,t}: H \to H$. Here $\AAA_{s,t}^*: H \to L^2([s,t])$ is the adjoint of $\AAA_{s,t}$ and satisfies
\begin{align}
    (\AAA_{s,t}^* \xi )(r) = \sigma \langle e_0, \JJ^*_{r,t} \xi  \rangle, \quad \textrm{ for } r \in [s,t],
    \label{eq:AAA:stared}
\end{align}
where $\JJ^*_{s,t}$ is the adjoint of $\JJ_{s,t}$ defined via \eqref{eq:rho:def:1}. $\JJ^*_{s,t}\xi$ solves the final value problem
\begin{align}
  -\frac{d}{dt} \rho^* + A\rho^* + (\nabla B(u))^* \rho^* = 0, \quad \rho^*(t) = \xi.
  \label{eq:J:star:back:eqn}
\end{align}
with the notation
\[
\nabla B(u) \rho = B(u,\rho) + B(\rho,u).
\] 

A formal solution of \eqref{eq:decay:cnt} is obtained by taking
$v = \AAA_{0,t}^*\MM_{0,t}^{-1}\JJ_{0,t} \xi$,
for some $t > 0$.  It is not expected however that $\MM_{0,t}$ is invertible for many infinite-dimensional problems.  This difficulty 
is circumvented by considering a regularization $\tilde{\MM}_{0,t}$ in place of $\MM_{0,t}$ so that the resulting control 
pushes $\rho$ into small scales (high wavenumbers). We then make use of the dissipative structure in 
\eqref{eq:shell:model:0}--\eqref{eq:shell:model:j} to induce a decay in $\rho$.
Specifically, we determine $v$ and the resulting controlled quantity $\rho$ according to the following 
iterative construction.  We start from $\rho(0) = \xi$ and, having determined $\rho$ and $v$ on an
interval $[0,2n]$ for some integer $n$, we define
\begin{align}
  v_{[2n, 2n+1]} = \AAA^*_{2n,2n+1}( \MM_{2n,2n+1} + \beta I)^{-1} \JJ_{2n,2n+1} \rho(2n),
  \quad \textrm{ and } \quad v_{[2n+1, 2n+2]} = 0.
  \label{eq:control:def}
\end{align}
Here $\beta$ is a fixed positive parameter that will be specified below according to \eqref{eq:very:hard:Mal:bbd}, \eqref{eq:2n:int:times:decay} and 
we have adopted the notation $v_{[s,t]}$ as the restriction of $v$ to the interval $[s,t]$.  With
$v$ now defined up to the time $2n+2$ we can then determine $\bar{\rho}$ on this interval via 
\begin{align}
  \bar{\rho}(t) =  
  \begin{cases}
  \JJ_{2n,t} \bar{\rho}(2n) - \AAA_{2n,t} v &  \textrm{ for } t \in [2n, 2n+1]\\
  \JJ_{2n+1, t} \bar{\rho}(2n+1) & \textrm{ for } t \in [2n+1, 2n+2).
  \end{cases}
    \label{eq:controled:quant:def}
\end{align}
Observe in particular that 
\begin{align}
  \bar{\rho}(2n+ 2) &= \JJ_{2n+1, 2n+2} \beta ( \MM_{2n,2n+1} + \beta I)^{-1} \JJ_{2n,2n+1}\bar{\rho}(2n).
  \label{eq:cont:err:two:step}
\end{align}
Note that $v$ and $\rho$ have a `block adapted' structure, that is, for each $t \geq 0$
\begin{align}
  \bar{\rho}(t), v(t) \textrm{ are } \mathcal{F}_{\varrho(t)} \textrm{- measurable}
  \label{eq:black:adapted:structure}
\end{align}
where, recalling the notation $\lceil t \rceil$ for the smallest integer greater
than or equal to $t$,
\begin{align*}
  \varrho(t) :=
  \begin{cases}
   \lceil t \rceil &\textrm{ when }  \lceil t \rceil \textrm{ is odd, }\\
   t &\textrm{ when }  \lceil t \rceil \textrm{ is even. }
  \end{cases}  
\end{align*}

\subsection{Decay estimates for $\bar{\rho}$}
\label{sec:decay:est:rho}
We next show how $v$ defined by \eqref{eq:control:def}, \eqref{eq:controled:quant:def} induces the desired decay \eqref{eq:decay:cnt}.
We start by demonstrating that for every $p > 1$, $n \geq 0$ and $\delta, \eta > 0$,
\begin{align}
   \E( |\bar{\rho}(2n +2)|^p | \mathcal{F}_{2n}) \leq \delta \exp(\eta |u(2n)|^2) |\bar{\rho}(2n)|^p
   \label{eq:one:time:step}
\end{align} 
holds for a suitably small choice of $0 < \beta = \beta(\delta, \eta, p)$, independent of $n$.
Splitting $\rho$ into low and high modes and using that $\|\beta (\MM_{2n, 2n+1} + \beta I)^{-1}\| \leq 1$ 
for any $\beta > 0$ we have\footnote{We use the notation $\|\cdot\|$ for the operator 
norm of bounded linear maps between the appropriate spaces
($H, L^2(s,t)$, etc.)}
\begin{align*}
 |\bar{\rho}(2n +2)|^p  
	&\leq C ( \|\JJ_{2n+1, 2n+2} Q_N\|^p + \| \JJ_{2n+1, 2n+2}\|^p \| P_N \beta (\MM_{2n,2n+1} +\beta I)^{-1}\|^p) \| \JJ_{2n,2n+1}\|^p |\bar{\rho}(2n)|^p\\
	&= (T_1 + T_2) |\bar{\rho}(2n)|^p
\end{align*}
which holds for any $n$ and every $\beta > 0$.  Since
\begin{align*}
  \E (T_1| \mathcal{F}_{2n}) \leq C\E (\E(\|\JJ_{2n+1, 2n+2} Q_N\|^p | \mathcal{F}_{2n+1})  \|  \JJ_{2n,2n+1}\|^p| \mathcal{F}_{2n})
\end{align*}
and
\begin{align*}
  \E (T_2| \mathcal{F}_{2n}) &\leq C\E (\E (\| \JJ_{2n+1, 2n+2}\|^p| \mathcal{F}_{2n+1}) \| P_N \beta (\MM_{2n,2n+1} +\beta I)^{-1}\|^p \| \JJ_{2n,2n+1}\|^p | \mathcal{F}_{2n}) ,
\end{align*}
the one step decay \eqref{eq:one:time:step} reduces to establishing that:
\begin{proposition}\label{prop:desired:bnds:1}
The following bounds hold:
\begin{itemize}
\item[(i)] For each $p > 1$ and each $\eta > 0$ we have
\begin{align}
  \E \|\JJ_{0,1}\|^p \leq C \exp( \eta |\IC|^2),
  \label{eq:J:one:time:step:est}
\end{align}
where the constant $C = C(\eta, p, \nu)$.
\item[(ii)] For all $q \geq 1$ and $\delta, \eta > 0$  there exists an $N$ such that
\begin{align}
	\E \|\JJ_{0,1} Q_N\|^q  \leq \delta \exp(\eta |\IC|^2)
  \label{eq:easer:FP:bnd}
\end{align}
where $Q_N$ is the projection onto $\mbox{span} \{e_0, \ldots, e_N\}^\perp$.
\item[(iii)] Finally, for every $q > 1$, $N> 0$ and $\eta, \delta > 0$ there exists $\beta > 0$ 
such that
\begin{align}
\E (\| P_N \beta (\MM_{0,1} +\beta I)^{-1}\|^q) \leq  \delta \exp(\eta |\IC |^2).
\label{eq:very:hard:Mal:bbd}
\end{align}
\end{itemize}
\end{proposition}

The first bound follows directly from \eqref{eq:rho:def:1} and \eqref{eq:exp:moment:decay:bnd}.  
The Foias-Prodi estimate \eqref{eq:easer:FP:bnd} expresses the fact that if an initial condition is 
concentrated in sufficiently high
wavenumbers then the diffusive terms in \eqref{eq:rho:def:1} mostly dissipates the 
solution after one time step.   The final bound \eqref{eq:very:hard:Mal:bbd} shows that inverting $\MM_{0,1} + \beta I$ approximately gives the desired control 
on the low modes. This step in the analysis is delicate and would not be expected to be true in general.  It relies on the fact that the
H\"ormander bracket condition, Proposition~\ref{prop:hor:cond} is satisfied.  We postpone further details for Sections~\ref{sec:FoiasProdi:bnds},~\ref{sec:mal:mat:analysis} below.

With \eqref{eq:one:time:step} in hand we establish \eqref{eq:decay:cnt} as follows. 
For any $q > 1$ and $\eta > 0$ define
\begin{align*}
  \mathfrak{P}_n := \prod_{k =1}^n  \left( \frac{| \bar{\rho}(2n +2) |}{| \bar{\rho}(2n)|} \right)^{q} \exp(-\eta/2 \cdot |u(2n) |^2) \quad \mbox{and} \quad
  \mathfrak{R}_n := \prod_{k =1}^n \exp(\eta/2 \cdot |u(2n)|^2).
\end{align*}
Note that $|\rho(2n+2)|^q :=  \mathfrak{P}_n\mathfrak{R}_n$. By making repeated use
of \eqref{eq:one:time:step}, we have that
 $(\E(  \mathfrak{P}_n\mathfrak{R}_n) )^{1/2}
  =   \E( \E (\mathfrak{P}_n^2 | \mathcal{F}_{2n})) \E ( \mathfrak{R}_n)^2
  \leq \delta \E (\mathfrak{P}_{n-1}^2) \E ( \mathfrak{R}_n)^2 
  \leq \cdots
  \leq \delta^n \E ( \mathfrak{R}_n)^2$.
On the other hand, from \eqref{eq:exp:moment:decay:bnd} we infer that
$\E \mathfrak{R}_n \leq  \exp(\eta |\IC|^2 + C_0 n)$
which is valid for sufficiently small $\eta = \eta(\nu) > 0$ and a constant $C_0 = C_0(\nu) > 0$. By taking $\delta = \exp(-2\gamma - C_0)$ 
in \eqref{eq:one:time:step} and combining these two bounds we now conclude 
\begin{align}
 \E (|\rho(2n+2)|^q) \leq \exp(\eta |\IC |^2 - 2n \gamma).
 \label{eq:2n:int:times:decay}
\end{align}
and hence \eqref{eq:decay:cnt}.

\subsection{Bounding the cost of control}
\label{sec:cost:control}
To obtain the cost of control bounds \eqref{eq:cost:cnt} we observe that
by using the {\em block adapted} structure in \eqref{eq:black:adapted:structure} with the generalized It\={o} isometry (see  \cite{Nualart2006}) we infer
\begin{align}
  \E \left( \int_0^{2n} v dW\right)^2 = \E \int_0^{2n} |v|^2 dt + \sum_{k= 0}^n\E \int_{2k}^{2k+1}\int_{2k}^{2k+1} \MD_s v(r) \MD_r v(s) dr ds.
  \label{eq:generalized:Ito:iso:block}
\end{align}
Here $\MD: \Mspc^p(H) \subset L^p(\Omega, H) \to L^p(\Omega; L^2([0,T]) \otimes H)$ is the Malliavin derivative operator.
For the first term in \eqref{eq:generalized:Ito:iso:block} observe that 
\begin{align}
  \E \int_0^{2n} |v|^2 ds 
   =& \sum_{k=0}^{n-1} \E \| \AAA^*_{2k, 2k+1} ( \MM_{2k, 2k+1} + \beta I)^{-1} \JJ_{2k, 2k+1} \rho(2k)\|_{L^2([2k, 2k+1])}^2 \notag\\
   \leq& \frac{1}{\beta}\sum_{k=0}^{n-1} \E (\| \JJ_{2k, 2k+1}\|^4)^{1/2}  (\E (|\rho(2k)  |^4))^{1/2} 
   \leq \frac{C \exp(\eta|\IC|^2)}{\beta}\sum_{k=0}^\infty \exp(-2\gamma k).
   \label{eq:Ito:Iso:term:1}
\end{align}
Here we have used that $\| \AAA^*_{2k, 2k+1} ( \MM_{2k, 2k+1} + \beta I)^{-1/2}\|_{\mathcal{L}(H, L^2([2k, 2k+1]))} \leq 1$ and that
$\| \MM_{2k, 2k+1} + \beta I)^{-1/2}\|\leq \beta^{-1/2}$.

In order to address the second term in \eqref{eq:generalized:Ito:iso:block} we use the (Malliavin) chain rule and the fact that 
$\rho_{2n}$ is $\mathcal{F}_{2n}$ adapted to compute
\begin{align}
  \MD_t v_{[2n,2n+1]} =& \MD_t \AAA^*_{2n,2n+1}( \MM_{2n,2n+1} + \beta I)^{-1} \JJ_{2n,2n+1} \rho(2n) \notag\\
  	  &+ \AAA^*_{2n,2n+1}  \MD_t ( \MM_{2n,2n+1} + \beta I)^{-1} \JJ_{2n,2n+1} \rho(2n) \notag\\
	  &+ \AAA^*_{2n,2n+1}( \MM_{2n,2n+1} + \beta I)^{-1}  \MD_t \JJ_{2n,2n+1} \rho(2n), 
	  \label{eq:MD:cont:mess:1}
\end{align}
for any $t \geq 2n$.
On the other hand 
\begin{align}
\MD_t &( \MM_{2n,2n+1} + \beta I)^{-1}  \label{eq:MD:cont:mess:2}\\
   &=  -( \MM_{2n,2n+1} + \beta I)^{-1} 
   	( \MD_t \AAA_{2n,2n+1} \AAA_{2n,2n+1}^* + \AAA_{2n,2n+1}  \MD_t\AAA_{2n,2n+1}^*) 
	( \MM_{2n,2n+1} + \beta I)^{-1}. \notag
\end{align}
In view of  \eqref{eq:MD:cont:mess:1}, \eqref{eq:MD:cont:mess:2}, we need more explicit expressions for 
$\MD_t \JJ_{2n, 2n+1}$, $\MD_t \AAA_{2n, 2n+1}$, and $\MD_t \AAA^*_{2n,2n+1}$.
For any $\xi, \xi' \in H$ we take $\tilde{\rho} = \JJ_{s,t}^{(2)}(\xi, \xi')$  as the solution of
 $\frac{d}{dt} \tilde{\rho} + \nu A \tilde{\rho} + B(u,\tilde{\rho}) + B(\tilde{\rho},u)  + 
 B(\JJ_{s,t} \xi ,\JJ_{s,t} \xi') + B(\JJ_{s,t} \xi',\JJ_{s,t} \xi) = 0$, $\tilde{\rho}(s) = 0.$
Using the properties $\MD_t$ one may show that (see \cite{HairerMattingly2011})
\begin{align}
  \MD_\tau \JJ_{s,t}\xi =
  \begin{cases}
     \JJ^{(2)}_{\tau,t}( \sigma e_0, \JJ_{s,\tau} \xi) & \textrm{ when } s < \tau,\\
     \JJ^{(2)}_{s,t}(\JJ_{\tau, s} \sigma e_0, \xi) &  \textrm{ when } s \geq \tau.
  \end{cases}
  \label{eq:MD:JJ}
\end{align}
By making use of \eqref{eq:MD:JJ} one may verify the
following additional moment bounds from \eqref{eq:exp:moment:decay:bnd},
\eqref{eq:rho:def:1}, \eqref{eq:A:op}, \eqref{eq:MD:JJ} and routine
estimations (see \cite{HairerMattingly06}).
\begin{lemma}
\label{lem:many:more:merry:moment:bnds}
\mbox{}
\begin{itemize}
\item[(i)]  
For any $T > 0$, $p \geq 1$, $\eta >0$
\begin{align*} 
 \E \sup_{t \in [0,T]} \|\JJ_{t,T}\|^p \leq C \exp(\eta |\IC|^2), \quad \E \sup_{t \in [0,T]} \|\JJ_{t,T}^{(2)}\|^p \leq C \exp (\eta |\IC|^2) 
\end{align*}
for a constant $C = C(T,p, \nu, \eta)$.  Similarly for $r < t$ and $p \geq 1$
\begin{align*}
   \E \|\AAA_{r,t}\|^p \leq C \exp(\eta |\IC |^2), \quad \E \|\AAA_{r,t}^*\|^p \leq C \exp(\eta |\IC |^2).
\end{align*}
\item[(ii)] For $r \leq s \leq t$, $p \geq 1$ and $\eta > 0$ we have
\begin{align*}
  \E \|\MD_s \JJ_{r, t}\|^p      \leq C \exp(\eta |\IC |^2), \quad
  \E \| \MD_s \AAA_{r, t} \|^p \leq C \exp(\eta |\IC|^2), \quad
  \E \| \MD_s \AAA^*_{r,t}\|^p \leq C \exp(\eta |\IC|^2),
\end{align*}
for a constant $C = C(p, t - r, \nu, \eta)$.
\end{itemize}
\end{lemma}
With these bounds in mind we now return to \eqref{eq:generalized:Ito:iso:block}. 
The second term in this expression is bounded by $\sum_{k=0}^{2n} \E \| \MD v \|_{L^2( [2k, 2k+1]^2)}^2$.
We handle each of the terms in this sum using the expression 
\eqref{eq:MD:cont:mess:1}, \eqref{eq:MD:cont:mess:2} as 
\begin{align}
   \| \MD v \|_{L^2( [2k, 2k+1]^2)}^2 \leq \frac{1}{\beta^2}
   \Big(&\|\MD_t \AAA^*_{2n,2n+1}\|^2 \|\JJ_{2n,2n+1}\|^2 +\|\MD_t \AAA_{2n,2n+1}\|^2 \|\JJ_{2n,2n+1}\|^2
   \notag\\
   &+ \| \AAA^*_{2n,2n+1}\|^2 \|\MD_t \JJ_{2n,2n+1}\|^2 \Big)
   |\rho(2n)|^2
   \label{eq:mall:der:term:ito}
\end{align}
where we have used that $\| \AAA^*_{2k, 2k+1} ( \MM_{2k, 2k+1} + \beta I)^{-1/2}\| \leq 1$,
$\|  ( \MM_{2k, 2k+1} + \beta I)^{-1/2} \AAA_{2k, 2k+1}\| \leq 1$,
and $\|  ( \MM_{2k, 2k+1} + \beta I)^{-1/2}\| \leq \beta^{-1/2}$.
Using \eqref{eq:2n:int:times:decay} and Lemma~\ref{lem:many:more:merry:moment:bnds}
with \eqref{eq:mall:der:term:ito} we conclude that
\begin{align}
    \sum_{k= 0}^n\E \int_{2k}^{2k+1}\int_{2k}^{2k+1} \MD_s v(r) \MD_r v(s) dr ds
    \leq \frac{\exp(\eta | \IC |^2)}{\beta^2} \sum_{k = 0}^n \exp( - \gamma k |\IC|^2).
\label{eq:Ito:Iso:term:2}
\end{align}
Combining \eqref{eq:Ito:Iso:term:1} and \eqref{eq:Ito:Iso:term:2} with \eqref{eq:generalized:Ito:iso:block} we conclude \eqref{eq:cost:cnt}. 

\subsection{Foias-Prodi-type bounds}
\label{sec:FoiasProdi:bnds}
We turn next to establishing \eqref{eq:easer:FP:bnd}, and prove \eqref{eq:J:one:time:step:est} along the way.
The importance of having a semi-linear system, ensured in our case by $1 \leq c < 2$, is directly apparent in the estimates of this section.
Recall the notation
$ 
\rho = \JJ_{s,t} \xi
$
for the linearized flow around the solution $u(t,\IC)\in C(0,T;H) \cap L^2(0,T;H^1) $ of \eqref{eq:abstract:shell};  that is, $\rho$ solves
\eqref{eq:rho:def:1}.

From the $L^2$ energy inequality and using \eqref{eq:B:cancelation}, \eqref{eq:B:bnds} we obtain
\begin{align*}
\frac{d}{dt} |\rho|^2 + 2 \nu |\rho|_{H^1}^2 
\leq 2 \left| \langle B(\rho,u),\rho \rangle \right| 
\leq 2 |\rho|_{H^{c-1}} |u|_{H^1} |\rho| 
\leq \nu |\rho|_{H^1}^2 + \nu^{- \frac{c-1}{3-c}} |\rho|^2 |u|_{H^1}^{\frac{2}{3-c}}
\end{align*}
for all $c \in [1,2]$. After absorbing the $\nu |\rho|_{H^1}^2$ term in the left hand side and multiplying the
resulting differential inequality by $|\rho|^{p-2}$ we infer
\begin{align*}
 \frac{d}{dt} |\rho|^{p} +  \frac{p \nu}{2}  |\rho|_{H^1}^2 | \rho |^{p-2} \leq \frac{p}{2}|\rho|^{p} \left( \nu^{-\frac{c-1}{3-c}} |u|_{H^1}^{\frac{2}{3-c}} \right) \leq  |\rho|^p \left( \kappa |u|_{H^1}^{2}  + C \right)
\end{align*}
for any $\kappa>0$ and $p \geq 2$ and a suitable constant $C = C(\nu,c,\kappa, p)$  that may be computed explicitly.  Note here that the final inequality requires that $1 \leq c < 2$. 
Letting $\kappa = \frac{\nu}{16 \sigma^2} \wedge \eta$, applying the Gr\"onwall inequality, taking expected values, and making use of \eqref{eq:exp:moment:decay:bnd} we arrive at
\begin{align}
\E |\rho(t)|^p + \frac{p \nu}{2} \int_0^t  \E \left( |\rho(s)|_{H^1}^2 |\rho(s)|^{p-2} \right)ds  \leq |\xi|^p  \exp\left( \eta |\IC|^{2}  + C t \right)
\label{eq:rho:moment:general:power}
\end{align}
for any $p\geq 2$, $t\geq 0$ where $C = C(\nu,\sigma,c,p)$.  The bound \eqref{eq:J:one:time:step:est} follows immediately.

Recall that $P_N$ is the projection onto the first $N$ coordinates of elements of $H$ 
and $Q_N = I - P_N$.  We denote by $\rho_l = P_N \rho$ and $\rho_h = Q_N \rho$ as
the {\em low} and the {\em high} components of $\rho$ solving \eqref{eq:rho:def:1}. Upon applying $Q_N$ to \eqref{eq:rho:def:1} we obtain
\begin{align*}
\partial_t \rho_h + A \rho_h + Q_N( B(u,\rho_l + \rho_h) + B(\rho_l + \rho_h, u) ) = 0.
\end{align*}
Multiplying with $\rho_h$, using that $2^{2N} | \rho_h|^2 \leq |\rho_h|^2_{H^1}$, the cancelation property \eqref{eq:B:cancelation}, 
and estimates similar to \eqref{eq:B:bnds} we obtain
\begin{align*}
\frac{d}{dt} |\rho_h|^2 + \nu 2^{2N} |\rho_h|^2 + \nu |\rho_h|_{H^1}^2 
&\leq 2 |\langle B(u,\rho_l), \rho_h \rangle|  + 2 |\langle B(\rho_l,u) \rho_h \rangle|  + 2 |\langle B(\rho_h,u), \rho_h \rangle|  \notag\\
&\leq 4 |u|_{H^1} |\rho_h| |\rho_l|^{2-c} |\rho_l|_{H^1}^{c-1} + 2 |u|_{H^1} |\rho_h|_{H^1}^{c-1} |\rho_h|^{3-c}.
\end{align*}
For $\kappa >0$ to be determined we infer that
\begin{align}
\frac{d}{dt} |\rho_h|^2 + \left( \nu 2^{2N} - \kappa |u|_{H^1}^2 \right) |\rho_h|^2 + \frac{\nu}{2} |\rho_h|_{H^1}^2 
	\leq C\left( |\rho_l |^{2(c-1)}_{H^1} |\rho|^{2(2-c)} + |\rho|^{2} \right)
	\leq 2^{2(c-1)N}C |\rho|^{2}.
\label{eq:rho_h:1}
\end{align}
where $C = C(\nu, \kappa, c)$ but is independent of $N$ and we have again used that $1 \leq c < 2$.
For any $p\geq 2$, upon multiplying \eqref{eq:rho_h:1} with $|\rho_h|^{p-2}$ and using the Gr\"onwall 
and H\"older inequalities we obtain
\begin{align} 
\E |\rho_h(t)|^p 
\leq&  |\xi|^p \E \left( \mu(t,0)^{\frac p2} \right) + 2^{2(c-1)N} C \int_0^t (\E \mu(t,s)^{p})^{1/2} (\E |\rho(s)|^{2p})^{1/2} ds
\label{eq:rho:high:gronwall}
\end{align}
where $C = C(\nu,\kappa,c,p)$, independent of $N$, and 
\begin{align*}
\mu(t,s) = \exp\left( -\nu 2^{2N} (t-s) + \kappa \int_s^t |u(\tau)|_{H^1}^2 d\tau \right).
\end{align*}
By letting $\kappa = p^{-1}(\frac{\nu^2}{16\sigma^2} \wedge \eta)$ and using \eqref{eq:exp:moment:decay:bnd} we have
\begin{align}
\E \mu(t,s)^p
&\leq 2 \exp\left( - \nu p 2^{2N}  (t-s) \right)  \exp\left(\eta |\IC|^2 \right),
\label{eq:E:mu:bound}
\end{align} 
for any $0 \leq s < t$.
Combining \eqref{eq:rho:moment:general:power}, \eqref{eq:rho:high:gronwall} with \eqref{eq:E:mu:bound} we obtain
\begin{align*} 
\E |\rho_h(t)|^p 
\leq&  \exp\left(\eta |\IC|^2 \right) \left(|\xi|^{p}     \exp\left( - \nu p 2^{2N-1} t \right) + 2^{2(c-1)N}   |\xi|^{2p} \frac{1}{  \nu p 2^{2N}}  \right)
\end{align*}
for a constant $C =C(\nu,\kappa,c,p,t)$ independent of $N$.  By now taking $t = 1$ and $N$ sufficiently large we now conclude \eqref{eq:easer:FP:bnd}.

\subsection{Analysis of the Malliavin covariance operator}
\label{sec:mal:mat:analysis}

The second crucial bound necessary to achieve Proposition~\ref{prop:cont:problem} is \eqref{eq:very:hard:Mal:bbd}.  
This inequality is immediately inferred from the following probabilistic spectral estimate on $\MM_{0,1}$ (see \cite{HairerMattingly2011}).
\begin{proposition}
\label{prop:mal:spec:bnd}
For every $\alpha, \gamma > 0$ and every integer $N$ there exists a $\delta > 0$ such 
that
\begin{align}
  \Prb \left(  \sup_{\xi \in \TaN} \frac{\langle \MM_{0,1}\xi, \xi \rangle}{|\xi|^2} < \epsilon \right)  \leq C\epsilon^\delta \exp(\gamma |\IC |^2)
  \label{eq:mall:spectral:bnd}
\end{align}
for every $\epsilon > 0$, where $\TaN :=  \{ \xi : |P_N \xi | \geq \alpha | \xi |\}$ and the constants $C= C(\alpha, \gamma, N)$ and $\delta= \delta(\alpha,\gamma, N) > 0$ 
are independent of $\epsilon$ and $\IC$.
\end{proposition}

The proof of the estimate \eqref{eq:mall:spectral:bnd} consists in translating 
each of the admissible brackets leading to the condition \eqref{eq:easiest:inf:dim:hormander} into 
quantitive bounds.   This leads to what amounts to an iterative proof by contradiction with high probability.  
One begins by showing that small eigenvalue, eigenvector pairs translate to a smallness
condition on linear forms related to successive Lie brackets as follows:

\begin{proposition} \label{prop:small:eigenval:imp:large:sets}
\mbox{}
\begin{itemize}
\item[(i)] There exists an $\epsilon_0 > 0$ and collection of measurable sets $\Omega_{\epsilon, 0}$ defined for each $\epsilon < \epsilon_0$ such
that $\Prb( \Omega_{\epsilon, 0}^C) \leq C\epsilon \exp(\eta |\IC|^2)$
and so that on $\Omega_{\epsilon, 0}$
\begin{align}
   \langle \MM_{0,1} \xi, \xi \rangle < \epsilon |\xi|^2 \quad \Rightarrow \quad \sup_{t \in [1/2, 1]} |\langle \JJ^*_{t, 1}\xi,  e_0 \rangle | < \epsilon^q  |\xi| ,
     \label{eq:initial:imp:good:set}
\end{align}
for every $\xi  \in H$.
\item[(ii)] Suppose that $E \in \mathcal{M}_k$, for some $k \geq 0$ where $\mathcal{M}_k$ is defined above in \eqref{eq:kth:order:iterated:poly}
and we take $\mathcal{M}_0 = \{ e_0 \}$.  
Then there exist $\epsilon_0 = \epsilon_0(E) >0$, $q = q(E)$ such that for every $\epsilon < \epsilon_0$
there is a set $\Omega_{\epsilon,E}$ so that $\Prb( \Omega_{\epsilon, E}^C) < C\epsilon \exp(\eta | \IC |^2)$
and so that on $\Omega_{\epsilon, E}$
\begin{align}
    &\sup_{t \in [1/2, 1]} |\langle \JJ^*_{t, 1} \xi, E(u) \rangle | < \epsilon |\xi |  \notag\\
     &\Rightarrow \; \left(\sup_{t \in [1/2, 1]} |\langle \JJ^*_{t, 1}\xi,  [E(u), F(u)] \rangle | +      \sup_{t \in [1/2, 1]} |\langle \JJ^*_{t, 1} \xi,  [E(u), e_0] \rangle |\right)  < \epsilon^q |\xi|,
     \label{eq:imp:iterative:bad:set}
\end{align}
for every $\xi \in H$.
\end{itemize}
\end{proposition}

The proof of Proposition~\ref{prop:small:eigenval:imp:large:sets} is lengthy and technical.  Here we merely hint at some details.   The complete 
proof follows exactly as in \cite{HairerMattingly2011} and see also \cite{FoldesGlattHoltzRichardsThomann2013}.
One obtains new brackets of the form $[E(u), e_0]$
by expanding $E(u) = E(\bar{u} + \sigma W)$ where $\bar{u} = u - \sigma W$ and then using a bound on Wiener polynomials from \cite{HairerMattingly2011} 
to show that each of the terms in the expansion is small if $E(u)$ is small.  Here may simplify the analysis by taking advantage of the smoothing estimate
\begin{align*}
   \E \sup_{t \in [t_0, t_1]} |u(t, \IC)|^p_{H^s}, \textrm{ for any } 0< t_0 < t_1 < \infty.
\end{align*}
Implications involving $[E(u), A(u) + B(u)]$ in \eqref{eq:imp:iterative:bad:set} are obtained by again
changing variables, differentiating in the expression $\langle \JJ^*_{t, 1} \xi, E(\bar{u}) \rangle$ and 
making use of interpolation bounds involving Holder regularity in time.

Iterating the chain of implications \eqref{eq:imp:iterative:bad:set} starting from \eqref{eq:initial:imp:good:set} we may
infer the smallness of any form associated with a sequence of admissible bracket operations; cf. Definition~\ref{def:Hormander:bracket}.
Thus Theorem~\ref{prop:hor:cond} and Proposition~\ref{prop:small:eigenval:imp:large:sets} imply
\begin{corollary}\label{cor:imp:plus:hor}
For every $N \geq 0$ there exists an $q = q(N) > 0$, $\epsilon_0 = \epsilon_0(N) > 0$ and 
sets $\Omega_\epsilon$ defined for $\epsilon \in [0, \epsilon_0]$
with
\begin{align}
   \Prb( \Omega_{\epsilon}^C) \leq \epsilon C \exp( \eta |\IC|^2) 
   \label{eq:final:good:sets}
\end{align}
and such that on $\Omega_{\epsilon}$ we have the implication
\begin{align}
    \langle \MM_{0,1} \xi, \xi \rangle < \epsilon |\xi|^2 
    \quad \Rightarrow \quad
    \sum_{k = 0}^N \langle \xi, e_k \rangle^2 \leq \epsilon^q |\xi|^2 
    \label{eq:imp:to:quad:form}
\end{align}
which holds for every $\xi \in H$.
\end{corollary}
We now infer Proposition~\ref{prop:mal:spec:bnd} from Corollary~\ref{cor:imp:plus:hor} as follows. Observe that for $\xi \in \TaN :=  \{ \xi : |P_N \xi| \geq \alpha |\xi|\}$
\begin{align*}
\alpha |\xi |^2 \leq | P_N \xi |^2 = \sum_{k = 0}^N \langle \xi, e_k \rangle^2.
\end{align*}
Therefore combining this bound with \eqref{eq:imp:to:quad:form} we infer that,
on the sets $\Omega_\epsilon$ given in \eqref{eq:final:good:sets}
we have that 
\begin{align*}
   \langle \MM_{0,1} \xi, \xi \rangle \geq \epsilon |\xi|^2 
\end{align*}
for every $\epsilon < \epsilon_1(N, \alpha)$ and each $\xi \in \TaN$.
This completes the proof of Proposition~\ref{prop:mal:spec:bnd}.

\subsection{Consequences of the Gradient Estimates}

We finally describe how Propositions~\ref{thm:irr:simple},~\ref{thm:grad:est:MSG} 
imply Theorem~\ref{thm:uniqueness:IM}.
In  \cite{HairerMattingly06, HairerMattingly2008} the authors show that in the general setting
of Markov semigroups on Banach spaces, 
the gradient bound in Proposition~\ref{thm:grad:est:MSG}, the irreducibility condition 
Proposition~\ref{thm:irr:simple}, and certain moment bounds satisfied by establishing \eqref{eq:exp:moment:decay:bnd}
imply the ergodicity and mixing properties of $\{P_t\}$ claimed in Theorem~\ref{thm:uniqueness:IM}.  The central limit theorem, (iii) 
follows from abstract results in \cite{KomorowskiWalczuk2012}.   Details of the application for the 
stochastic Navier-Stokes equations are given in these works and are precisely the same in our
situation. See also \cite{FoldesGlattHoltzRichardsThomann2013} where these results are shown to apply to a 
different concrete infinite dimensional stochastic system.

We prove the strong law of large numbers \eqref{eq:SLLN} following
the strategy taken in \cite{KuksinShirikian12}.
This requires some suitable modifications to the proof however since mixing occurs in a weaker
sense, \eqref{eq:mixing:wash:norm}, than in \cite{KuksinShirikian12} where only a non-degenerate 
stochastic forcing is considered.

We will consider, without loss of generality, that $\int \phi(u) d \mu(u) = 0$. 
The proof \eqref{eq:SLLN} relies on the stochastic process
\begin{align*}
  M_{T} = \int_0^\infty \left( \E(\phi(u(t, \IC))| \mathcal{F}_T) -  \E \phi(u(t, \IC))  \right) dt
\end{align*}
Observe that, with the Markov property,
\begin{align}
  M_T &= \int_0^T \phi(u(t,\IC)) dt + \int_0^\infty P_{t} \phi(u(T, \IC))dt - \int_0^\infty P_t\phi(\IC) dt \notag\\
          &:= \int_0^T \phi(u(t,\IC)) dt + R(u(T, \IC)) - R(\IC).
          \label{eq:MG:Decomposition}
\end{align}
We establish the convergence \eqref{eq:SLLN} using $M_T$ in two steps.  Firstly we 
show
\begin{align}
   \frac{ R(u(T,\IC)) - R(\IC)}{T} = \frac{1}{T}\left( \int_0^T \phi(u(t, \IC)) dt - M_T\right) \to 0  \quad a.s. 
   \label{eq:MG:Ave:diff}
\end{align}
and then we establish that
\begin{align}
   \frac{M_T}{T} \to 0 \quad a.s.
   \label{eq:weird:MG:to:zero}
\end{align}

For the first convergence, \eqref{eq:MG:Ave:diff}, we infer from \eqref{eq:mixing:wash:norm} that 
\begin{align*}
  \frac{R(u(T, \IC))}{T} \leq \frac{C \exp(\eta/2 | u(T, \IC)|^2)}{T}.
\end{align*}
To show that the later quantity goes to zero fix any $\delta > 0$, 
and observe that
\begin{align*}
  \sum_{N \geq 1} \Prb\left(  \frac{ \exp(\eta/2 | u(\delta N, \IC)|^2)}{\delta N} \geq  N^{-1/4} \right) 
  \leq \frac{1}{\epsilon^2 \delta^2} \sum_{N \geq 1} \frac{\E \exp(\eta | u(\delta N, \IC)|^2)}{N^{3/2}}.
\end{align*}
With the Borel-Cantelli lemma we infer that,
\begin{align*}
  \bigcup_{M= 1}^\infty \left\{\frac{ \exp(\eta/2 | u(\delta N, \IC)|^2)}{\delta N} <  \frac{1}{N^{1/4}}, \textrm{ for every } N \geq M \right\} 
\end{align*}
has measure one.   Since this holds for all $\delta > 0$ we infer the first convergence \eqref{eq:MG:Ave:diff}.

We turn to the second convergence \eqref{eq:weird:MG:to:zero} which we address with the strong law of large numbers for Martingales.  
Making use of \eqref{eq:mixing:wash:norm} we observe, for suitable $\gamma_1, \gamma_2$ that
\begin{align}
  \E R(u(T, \IC))^2  
  \leq& C\E \left( \int_0^\infty  \exp( - \gamma_1 t + \eta/2 | u(T, \IC)|^2) \| \phi\|_{\gamma_2} dt \right)^2
  \leq C \E \exp(\eta | u(T, \IC)|^2) \notag\\
  \leq& C \exp(\eta | \IC|^2)
  \label{eq:R:exp:bnd}
\end{align}
where $C$ does not depend on $T$ and where we have used \eqref{eq:exp:moment:decay:bnd}
for the final bound. Similar bounds apply for $R(\IC)$ for the same reasons.
With this bound in hand it is direct to verify that $\{M_{T}\}_{T \geq 0}$ is a square integrable, mean zero martingale.
It is therefore sufficient to show that for $\delta >0$, 
\begin{align}
\sum_{N \geq 1} \frac{\E (M_{\delta N} - M_{\delta (N-1)})^2}{N^2} < \infty,
\label{eq:mg:sum}
\end{align}
see for example \cite{KuksinShirikian12}.
Using the bound \eqref{eq:R:exp:bnd} we have
\begin{align}
  \E (M_{\delta N} - M_{\delta (N-1)})^2 
  =& \E \left(  \int_{\delta (N-1)}^{\delta N} \phi(u(t, \IC)) dt + R(u(\delta N, \IC)) - R(u(\delta (N-1), \IC)) \right)^2 \notag\\
  \leq& C \left(\delta \int_{\delta (N-1)}^{\delta N} \E \phi(u(t, \IC))^2 dt +  \exp(\eta | \IC|^2) \right),
  \label{eq:incre:bnd}
\end{align}
for a constant $C$ independent of $\delta$, $N$.
Now, since $\phi \in \mathcal{G}$ it is easy to see that $\phi^2 \in \mathcal{G}$; cf. \eqref{eq:exp:grow:space}.
We there infer from 
\begin{align}
  \E \phi^2(u(t, \IC)) \leq C + \int \phi^2(u) d \mu(u)
  \label{eq:mixing:bnd:phi:sq}
\end{align}
where the constant $C = C(\eta, c , \sigma, \phi)$ is independent of $t$.  Combining \eqref{eq:incre:bnd} and \eqref{eq:mixing:bnd:phi:sq}
we infer \eqref{eq:mg:sum} and hence, since $\delta > 0$ is arbitrary,  \eqref{eq:weird:MG:to:zero} follows.

\bigskip
\begin{center}
{\bf Acknowledgments}
\end{center}
We would like to thank the Institut Henri Poincar\'e where this work was conceived during a summer 
Research in Paris grant.  We would also like to acknowledge the Newton Institute at the University of Cambridge
where NEGH was a visitor during the final stages of writing.  We thank Marco Romito for the helpful feedback and references.  
This work was also partially supported under the grants NSF DMS-1207780 (SF), NSF DMS-1313272 (NEGH),
NSF DMS-1211828 (VV).


\end{document}